\documentclass[11pt]{article}
\usepackage[letterpaper,hmargin=1in,vmargin=1.25in]{geometry}
\usepackage[english]{babel}
\usepackage[utf8]{inputenc}
\usepackage[pdftex,
            pdfauthor={Pablo Rauzy, Sylvain Guilley, Zakaria Najm},
            colorlinks=false,pdfborder={0 0 0},
            pdftitle={Formally Proved Security of Assembly Code Against Power Analysis}
            ]{hyperref}
\usepackage{graphicx}
\usepackage{amsfonts,amssymb,mathtools,epstopdf}
\usepackage{color,xspace,xcolor,paralist,wrapfig}
\usepackage{relsize,alltt,multirow,booktabs,caption,subcaption,moreverb}
\usepackage{amsthm}
\theoremstyle{definition}
\newtheorem{definition}{Definition}
\newtheorem{remark}{Remark}
\theoremstyle{plain}
\newtheorem{proposition}{Proposition}

\usepackage{xmpincl}
\includexmp{CC_Attribution_4.0_International}

\usepackage{rotating}

\usepackage[printonlyused]{acronym}
\newcommand{\aci}[1]{\acfi{#1}\acused{#1}}

\makeatletter
\def\thm@space@setup{%
  \thm@preskip=1mm \thm@postskip=1mm
}
\makeatother

\newcommand{\footref}[1]{$^{\ref{#1}}$}

\newcommand{\ie}{\textit{i.e.}}
\newcommand{\eg}{\textit{e.g.}}
\newcommand{\etal}{\textit{et al.}\xspace}
\newcommand{\etc}{\textit{etc.}\xspace}

\newcommand{\present}{{\sc present}\xspace}

\newcommand{\paioli}{{\sf paioli}\xspace}

\title{
  {\bf Formally Proved Security of Assembly Code Against Power Analysis}\\
  {\smaller A Case Study on Balanced Logic}
}

\date{}

\author{
  Pablo Rauzy
  \and
  Sylvain Guilley
  \and
  Zakaria Najm\\
  Institut Mines-Télécom ; Télécom ParisTech ; CNRS LTCI\\
  {\it firstname}{\tt .}{\it lastname}{\tt @telecom-paristech.fr}
}

\begin{document}
\maketitle

\begin{abstract}
In his keynote speech at CHES 2004, Kocher advocated that side-channel attacks were an illustration that formal cryptography was not as secure as it was believed because some assumptions (\eg, no auxiliary information is available during the computation) were not modeled.
This failure is caused by formal methods' focus on models rather than implementations.
In this paper we present formal methods and tools for designing protected code and proving its security against power analysis.
These formal methods avoid the discrepancy between the model and the implementation by working on the latter rather than on a high-level model.
Indeed, our methods allow us
\begin{inparaenum}[(a)]
\item to automatically insert a power balancing countermeasure directly at the assembly level, and to prove the correctness of the induced code transformation; and
\item to prove that the obtained code is balanced with regard to a reasonable leakage model.
\end{inparaenum}
We also show how to characterize the hardware to use the resources which maximize the relevancy of the model.
The tools implementing our methods are then demonstrated in a case study on an $8$-bit AVR smartcard for which we generate a provably protected \present implementation that reveals to be at least $250$ times more resistant to CPA attacks.
\end{abstract}

\noindent \textbf{Keywords.}
Dual-rail with Precharge Logic (DPL), formal proof, static analysis, symbolic execution, implementation, DPA, CPA, smartcard, PRESENT, block cipher, Hamming distance, OCaml.

\section{Introduction}
\label{formaldpl:intro}

The need to trust code is a clear and proved fact, but the code itself needs to be proved before it can be trusted.
In applications such as cryptography or real-time systems, formal methods are used to prove functional properties on the critical parts of the code.
Specifically in cryptography, some non-functional properties are also important, but are not typically certified by formal proofs yet.
One example of such a property is the resistance to side-channel attacks.
Side-channel attacks are a real world threat to cryptosystems;
they exploit auxiliary information gathered from implementations through physical channels such as power consumption, electromagnetic radiations, or time, in order to extract sensitive information (\eg, secret keys).
The amount of leaked information depends on the implementation and as such appears difficult to model.
As a matter of fact, physical leakages are usually not modeled when it comes to prove the security properties of a cryptographic algorithm.
By applying formal methods directly on implementations we can avoid the discrepancy between the model and the implementation.
Formally proving non-functional security properties then becomes a matter of modeling the leakage itself.
In this chapter we make a first step towards formally trustable cryptosystems, including for non-functional properties, by showing that modeling leakage and applying formal methods to implementations is feasible.

Many existing countermeasures against side-channel attacks are implemented at the hardware level, especially for smartcards.
However, software level countermeasures are also very important,
not only in embedded systems where the hardware cannot always be modified or updated,
but also in the purely software world.
For example, Zhang \etal~\cite{DBLP:conf/ccs/ZhangJRR12} recently extracted private keys using side-channel attacks against a target virtual machine running on the same physical server as their virtual machine.
Side channels in software can also be found each time there are some non-logic behaviors (in the sense that it does not appear in the equations / control-flow modeling the program) such as timing or power consumption (refer to~\cite{kocher-dpa_and_related_attacks}), but also some software-specific information such as packet size for instance (refer to~\cite{DBLP:journals/jce/MatherO12}).

In many cases where the cryptographic code is executed on secure elements (smartcards, TPM, tokens, \etc) side-channel and fault analyses are the most natural attack paths.
A combination of signal processing and statistical techniques on the data obtained by side-channel analysis allows to build key hypotheses distinguishers.
The protection against those attacks is necessary to ensure that secrets do not leak, and most secure elements are thus evaluated against those attacks.
Usual certifications are the common criteria (ISO/IEC 15408), the FIPS 140-2 (ISO/IEC 19790), or proprietary schemes (EMVCo, CAST, \etc).

\paragraph{Power analysis.}
It is a form of side-channel attack in which the attacker measures the power consumption of a cryptographic device.
\aci{SPA} consists in directly interpreting the electrical activity of the cryptosystem.
On unprotected implementations it can for instance reveal the path taken by the code at branches even when timing attacks~\cite{kocher-timing_attacks} cannot.
\aci{DPA}~\cite{kocher-dpa_and_related_attacks} is more advanced: the attacker can compute the intermediate values within cryptographic computations by statistically analyzing data collected from multiple cryptographic operations.
It is powerful in the sense that it does not require a precise model of the leakage, and thus works blind, \ie, even if the implementation is blackbox.
As suggested in the original \ac{DPA} paper by Kocher \etal~\cite{kocher-dpa_and_related_attacks}, power consumption is often modeled by Hamming weight of values or Hamming distance of values' updates as those are very correlated with actual measures.
Also, when the leakage is little noisy and the implementation is software,
\aci{ASCA}~\cite{DBLP:conf/cisc/RenauldS09} are possible;
they consist in modelling the leakage by a set of Boolean equations, where the key bits are the only unknown variables~\cite{DBLP:journals/jce/CarletFGR12}.
\smallskip

Thwarting side-channel analysis is a complicated task, since an unprotected implementation leaks at every step.
Simple and powerful attacks manage to exploit any bias.
In practice, there are two ways to protect cryptosystems: ``palliative'' versus ``curative'' countermeasures.
Palliative countermeasures attempt to make the attack more difficult, however without a theoretical foundation.
They include variable clock, operations shuffling, and dummy encryptions among others (see also~\cite{DBLP:conf/ches/GuneysuM11}).
The lack of theoretical foundation make these countermeasures hard to formalize and thus not suitable for a safe certification process.
Curative countermeasures aim at providing a leak-free implementation based on a security rationale.
The two defense strategies are
\begin{inparaenum}[(a)]
\item make the leakage as decorrelated from the manipulated data as possible (\emph{masking}~\cite[Chp.~9]{dpa_book}), or
\item make the leakage constant, irrespective of the manipulated data (hiding or \emph{balancing}~\cite[Chp.~7]{dpa_book}).
\end{inparaenum}

\paragraph{Masking.}
Masking mixes the computation with random numbers, to make the leakage (at least in average) independent of the sensitive data.
Advantages of masking are (\emph{a priori})
the independence with respect to the leakage behavior of the hardware,
and the existence of provably secure masking schemes~\cite{DBLP:conf/ches/RivainP10}.
There are two main drawbacks to masking.
First of all, there is the possibility of high-order attacks (that examine the variance or the joint leakage);
when the noise is low, \acp{ASCA} can be carried out on one single trace~\cite{DBLP:conf/ches/RenauldSV09},
despite the presence of the masks, that are just seen as more unknown variables, in addition to the key.
Second, masking demands a greedy requirement for randomness (that is very costly to generate).
Another concern with masking is the overhead it incurs in the computation time.
For instance, a provable masking of \acs{AES}-128 is reported in~\cite{DBLP:conf/ches/RivainP10} to be $43$ (resp. $90$) times slower than the non-masked implementation with a $1$st (resp. $2$nd) order masking scheme.
Further, recent studies have shown that masking cannot be analyzed independently from the execution platform:
for example \emph{glitches} are transient leakages that are likely to depend on more than one sensitive data, hence being high-order~\cite{mangard-ches06}.
Indeed, a glitch occurs when there is a race between two signals, \ie, when it involves more than one sensitive variable.
Additionally, the implementation must be carefully scrutinized to check for the absence of \emph{demasking} caused by overwriting a masked sensitive variable with its mask. 

\paragraph{Balancing.}
Balancing requires a close collaboration between the hardware and the software: two indistinguishable resources, from a side-channel point of view, shall exist and be used according to a dual-rail protocol.
\aci{DPL} consists in precharging both resources, so that they are in a common state, and then setting one of the resources.
Which resource has been set is unknown to the attacker, because both leak in indistinguishable ways (by hypothesis).
This property is used by the \ac{DPL} protocol to ensure that computations can be carried out without exploitable leakage~\cite{DBLP:journals/tcad/TiriV06}.

\paragraph{Contributions.}
\acf{DPL} is a simple protocol that may look easy to implement correctly;
however,
in the current context of awareness about cyber-threats, it becomes evident that (independent) formal tools that are able to \emph{generate} and \emph{verify} a ``trusted'' implementation have a strong value.
\begin{compactitem}
\item We describe a design method for developing balanced assembly code by making it obey the \ac{DPL} protocol.
This method consists in automatically inserting the countermeasure and formally proving that the induced code transformation is correct (\ie, semantic preserving).
\item We present a formal method (using symbolic execution) to statically prove the absence of power consumption leakage in assembly code provided that the hardware it runs on satisfies a finite and limited set of requirements corresponding to our leakage model.
\item We show how to characterize the hardware to run the \ac{DPL} protocol on resources which maximize the relevancy of the leakage model.
\item We provide a tool called \paioli\footnote{\label{formaldpl:fn:paioli}\url{http://pablo.rauzy.name/sensi/paioli.html}} which implements the automatic insertion of the \ac{DPL} countermeasure in assembly code, and, independently, is able to statically prove the power balancing of a given assembly code.
\item Finally, we demonstrate our methods and tool in a case study on a software implementation of the \present~\cite{DBLP:conf/ches/BogdanovKLPPRSV07} cipher running on an $8$-bit AVR micro-controller.
Our practical results are very encouraging: the provably balanced \ac{DPL} protected implementation is \emph{at least} $250$ times more resistant to power analysis attacks than the unprotected version while being only $3$ times slower.
The \ac{SNR} of the leakage is divided by approximately $16$.
\end{compactitem}

\paragraph{Related work.}
The use of formal methods is not widespread in the domain of implementations security.
In cases where they exist, security proofs are usually done on mathematical models rather than implementations.
An emblematic example is the Common Criteria~\cite{cc}, that bases its ``formal'' assurance evaluation levels on ``Security Policy Model(s)'' (class SPM) and not on implementation-level proofs.
This means that it is the role of the implementers to ensure that their implementations fit the model, which is usually done by hand and is thus error-prone.
For instance, some masking implementations have been proved;
automatic tools for the insertion of masked code have even been prototyped~\cite{DBLP:conf/ches/MossOPT12}.
However, masking relies a lot on randomness, which is a rare resource and is hard to formally capture.
Thus, many aspects of the security are actually displaced in the randomness requirement rather that soundly proved.
Moreover, in the field of masking, most proofs are still literate (\ie, verified manually, not by a computer program).
This has led to a recent security breach in what was supposed to be a proved~\cite{DBLP:conf/ches/RivainP10} masking implementation~\cite{DBLP:conf/fse/CarletGPQR12}.
Previous similar examples exist, \eg, the purported high-order masking scheme~\cite{DBLP:conf/ctrsa/SchrammP06},
defeated one year after in~\cite{DBLP:conf/ches/CoronPR07}.

Timing and cache attacks are an exception as they benefit from the work of Köpf \etal~\cite{DBLP:conf/ccs/KopfB07,DBLP:conf/csfw/KopfD09}.
Their tool, CacheAudit~\cite{DBLP:journals/iacr/DoychevFKMR13}, implements formal methods that directly work on x$86$ binaries.

Since we started our work on \ac{DPL}, others have worked on similar approaches.
Independently, it has been shown that \ac{SNR} reduction is possible with other encodings that are less costly, such as ``dual-nibble'' (Chen \etal~\cite{CC:CARDIS14}) or ``m-out-of-n'' (Servant \etal~\cite{VS:CARDIS14}).
However, it becomes admittedly much more difficult to balance the
resources aimed at hiding one each other.
Thus, there is a trade-off between performance (in terms of execution speed and code size) and security.
In this chapter we propose a proof-of-concept of maximal security.

In this light it is easy to conclude that the use of formal methods to prove the security of implementations against power analysis is a need, and a technological enabler: it would guarantee that the instantiations of security principles are as strong as the security principles themselves.

\paragraph{Organization of the chapter.}
The \ac{DPL} countermeasure is studied in Sec.~\ref{formaldpl:dpl}.
Sec.~\ref{formaldpl:paioli-gen} details our method to balance assembly code and prove that the proposed transformation is correct.
Sec.~\ref{formaldpl:paioli-proof} explains the formal methods used to compute a proof of the absence of power consumption leakage.
Sec.~\ref{formaldpl:avr} is a practical case study using the \present algorithm on an AVR micro-controller.
Conclusions and perspectives are drawn in Sec.~\ref{formaldpl:concl}.

\section{\acl{DPL}}
\label{formaldpl:dpl}

Balancing (or hiding) countermeasures have been employed against side channels since early 2004, with dual-rail with precharge logic.
The \ac{DPL} countermeasure consists in computing on a redundant representation: each bit $y$ is implemented as a pair $(y_\text{False}, y_\text{True})$.
The bit pair is then used in a protocol made up of two phases:
\begin{compactenum}
\item a \emph{precharge} phase, during which all the bit pairs are zeroized $(y_\text{False}, y_\text{True})=(0,0)$, such that the computation starts from a known reference state;
\item an \emph{evaluation} phase, during which the $(y_\text{False}, y_\text{True})$ pair is equal to $(1,0)$ if it carries the logical value $0$, or $(0,1)$ if it carries the logical value $1$.
\end{compactenum}
The value $(y_\text{False}, y_\text{True})=(1,1)$ is unused.
As suggested in~\cite{micromicro2003}, it can serve as a \emph{canary} to detect a fault.
Besides, if a fault turns a $(1,0)$ or $(0,1)$ value into either $(0,0)$ or $(1,1)$, then the previous functional value has been forgiven.
It is a type of infection, already mentioned in~\cite{DBLP:conf/eurocrypt/IshaiPSW06,NS:FDTC09}.
Unlike other infective countermeasure, \ac{DPL} is not scary~\cite{DBLP:conf/fdtc/BattistelloG13}, in that it consists in an erasure.
Indeed, the mutual information between the erased and the initial data is zero
(provided only one bit out of a dual pair is modified).

\subsection{State of the Art}
\label{formaldpl:dpl-sota}

\begin{figure}
  \centering
  \includegraphics[width=0.7\textwidth]{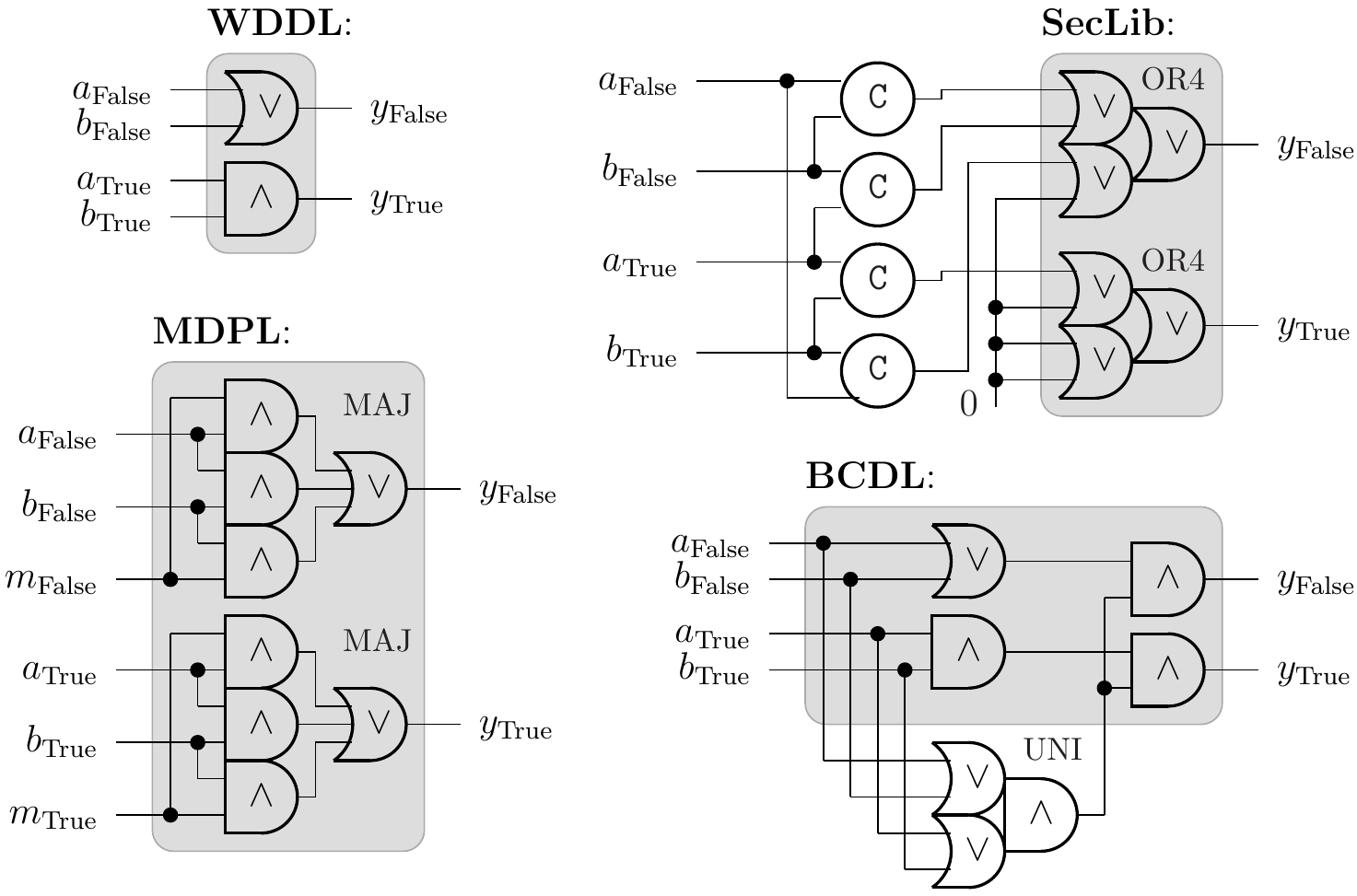}
  \caption{\label{formaldpl:fig-dpl_logic_styles} Four dual-rail with precharge logic styles.}
\end{figure}

Various \ac{DPL} styles for electronic circuits have been proposed.
Some of them, implementing the same logical {\it and} functionality, are represented in Fig.~\ref{formaldpl:fig-dpl_logic_styles}; many more variants exist, but these four are enough to illustrate our point.
The reason for the multiplicity of styles is that the indistinguishability hypothesis on the two resources holding $y_\text{False}$ and $y_\text{True}$ values happens to be violated for various reasons, which leads to the development of dedicated hardware.
A first asymmetry comes from the gates driving $y_\text{False}$ and $y_\text{True}$.
In \aci{WDDL}~\cite{tiri-date04}, these two gates are different: logical {\it or} versus logical {\it and}.
Other logic styles are balanced with this respect.
Then, the load of the gate shall also be similar.
This can be achieved by careful place-and-route constraints~\cite{tiri-cardis04,guilley-ches05}, that take care of having lines of the same length, and that furthermore do not interfere one with the other (phenomenon called ``crosstalk'').
As those are complex to implement exactly for all secure gates, the \aci{MDPL}~\cite{LNCS36590172} style has been proposed: instead of balancing exactly the lines carrying $y_\text{False}$ and $y_\text{True}$, those are randomly swapped, according to a random bit, represented as a pair $(m_\text{False}, m_\text{True})$ to avoid it from leaking.
Therefore, in this case, not only the computing gates are the same ({\it viz.} a majority), but the routing is balanced thanks to the extra mask.
However, it appeared that another asymmetry could be fatal to \ac{WDDL} and \ac{MDPL}: the gates pair could evaluate at different dates, depending on their input.
It is important to mention that side-channel acquisitions are very accurate in timing (off-the-shelf oscilloscopes can sample at more than $1$\,Gsample/s, \ie, at a higher rate than the clock period), but very inaccurate in space (\ie, it is difficult to capture the leakage of an area smaller than about $1$\,mm$^2$ without also recording the leakage from the surrounding logic).
Therefore, two bits can hardly be measured separately.
To avoid this issue, every gate has to include some synchronization logic.
In Fig.~\ref{formaldpl:fig-dpl_logic_styles}, the ``computation part'' of the gates is represented in a grey box.
The rest is synchronization logic.
In SecLib~\cite{SG:TC08}, the synchronization can be achieved by Muller C-elements (represented with a symbol {\tt C}~\cite{shams-cmos_cmuller}), and act as a decoding of the inputs configuration.
Another implementation, \aci{BCDL}~\cite{MN:DATE10}, parallelize the synchronization with the
computation.

\subsection{\acs{DPL} in Software}
\label{formaldpl:dpl-sw}

In this chapter, we want to run \ac{DPL} on an off-the-shelf processor.
Therefore, we must:
\begin{inparaenum}[(a)]
\item identify two similar resources that can hold true and false values in an indiscernible way for a side-channel attacker;
\item play the \ac{DPL} protocol by ourselves, in software.
\end{inparaenum}
We will deal with the former in Sec.~\ref{formaldpl:paioli-hw}. The rest of this section deals with the latter.

The difficulty of balancing the gates in hardware implementations is simplified in software.
Indeed in software there are less resources than the thousands of gates that can be found in hardware (aimed at computing fast, with parallelism).
Also, there is no such problem as early evaluation, since the processor executes one instruction after the other; therefore there are no unbalanced paths in timing.
However, as noted by Hoogvorst \etal~\cite{hoogvorstDD11}, standard micro-processors cannot be used {\it as is} for our purpose:
instructions may clobber the destination operand without precharge;
arithmetic and logic instructions generate numbers of $1$ and $0$ which depend on the data.

\begin{wrapfigure}{R}{0.25\textwidth} 
  \centering
  \begin{tabular}{r c l}
  $r_1$ & $\leftarrow$ & $r_0$ \\
  $r_1$ & $\leftarrow$ & $a$ \\
  $r_1$ & $\leftarrow$ & $r_1 \wedge 3$ \\
  $r_1$ & $\leftarrow$ & $r_1 \ll 1$ \\
  $r_1$ & $\leftarrow$ & $r_1 \ll 1$ \\
  $r_2$ & $\leftarrow$ & $r_0$ \\
  $r_2$ & $\leftarrow$ & $b$ \\
  $r_2$ & $\leftarrow$ & $r_2 \wedge 3$ \\
  $r_1$ & $\leftarrow$ & $r_1 \vee r_2$ \\
  $r_3$ & $\leftarrow$ & $r_0$ \\
  $r_3$ & $\leftarrow$ & $op[r_1]$ \\
  $d$   & $\leftarrow$ & $r_0$ \\
  $d$   & $\leftarrow$ & $r_3$ \\
  \end{tabular}
  \caption{\label{formaldpl:fig-dpl_macro} \protect\centering \ac{DPL} macro for $d = a~\mathrm{op}~b$.}
\end{wrapfigure}

To reproduce the \ac{DPL} protocol in software requires
\begin{inparaenum}[(a)]
\item \label{formaldpl:dplsw-req-bit} to work at the bit level, and
\item \label{formaldpl:dplsw-req-dup} to duplicate (in positive and negative logic) the bit values.
\end{inparaenum}
Every algorithm can be transformed so that all the manipulated values are bits (by the theorem of equivalence of universal Turing machines), so ({\it \ref{formaldpl:dplsw-req-bit}}) is not a problem.
Regarding ({\it \ref{formaldpl:dplsw-req-dup}}), the idea is to use two bits in each register / memory cell to represent the logical value it holds.
For instance using the two least significant bits, the logical value $1$ could be encoded as $1$ ({\tt 01}) and the logical value $0$ as $2$ ({\tt 10}).
Then, any function on those bit values can be computed by a look-up table indexed by the concatenation of its operands.
Each sensitive instruction can be replaced by a \emph{\ac{DPL} macro} which does the necessary precharge and fetch the result from the corresponding look-up table.

Fig.~\ref{formaldpl:fig-dpl_macro} shows a \ac{DPL} macro for the computation of $d = a~\mathrm{op}~b$, using the two least significant bits for the \ac{DPL} encoding.
The register $r_0$ is an always-zero register, $a$ and $b$ hold one \ac{DPL} encoded bit, and $op$ is the address in memory of the look-up table for the op operation.

This \ac{DPL} macro assumes that before it starts the state of the program is a valid \ac{DPL} state (\ie, that $a$ and $b$ are of the form {\tt /.+(01|10)/}\footnote{As a convenience, we use regular expressions notation.}) and leaves it in a valid \ac{DPL} state to make the macros chainable.

The precharge instructions (like $r_1 \leftarrow r_0$) erase the content of their destination register or memory cell before use.
If the erased datum is sensitive it is \ac{DPL} encoded, thus the number of bit flips (\ie, the Hamming distance of the update) is independent of the sensitive value.
If the erased value is not sensitive (for example the round counter of a block cipher) then the number of bit flips is irrelevant.
In both cases the power consumption provides no sensitive information.

The activity of the shift instructions (like $r_1 \leftarrow r_1 \ll
1$) is twice the number of \ac{DPL} encoded bits in $r_1$ (and thus does not
depend on the value when it is \ac{DPL} encoded).
The two most significant bits are shifted out and must be $0$, \ie, they cannot encode a \ac{DPL} bit.
The logical {\it or} instruction (as in $r_1 \leftarrow r_1 \vee r_2$) has a constant activity of one bit flip due to the alignment of its operands.
The logical {\it and} instructions (like $r_1 \leftarrow r_1 \wedge 3$) flips as many bits as there are $1$s after the two least significant bits (it's normally all zeros).

Accesses from/to the RAM (as in $r_3 \leftarrow op[r_1]$) cause as many bit flips as there are $1$s in the transferred data, which is constant when \ac{DPL} encoded.
Of course, the position of the look-up table in the memory is also important.
In order not to leak information during the addition of the offset ($op + r_1$ in our example), $op$ must be a multiple of $16$ so that its four least significant bits are $0$ and the addition only flips the number of bits at $1$ in $r_1$, which is constant since at this moment $r_1$ contains the concatenation of two \ac{DPL} encoded bit values.

We could use other bits to store the \ac{DPL} encoded value, for example the least and the third least significant bits.
In this case $a$ and $b$ have to be of the form {\tt /.+(0.1|1.0)/}, only one shift instruction would have been necessary, and the {\tt and} instructions' mask would be $5$ instead on $3$.

\section{Generation of \acs{DPL} Protected Assembly Code}
\label{formaldpl:paioli-gen} 

Here we present a generic method to protect assembly code against power analysis.
To achieve that we implemented a tool (See App.~\ref{sensi:paioli}) which transforms assembly code to make it compliant with the \ac{DPL} protocol described in Sec.~\ref{formaldpl:dpl-sw}.
To be as universal as possible the tool works with a generic assembly language presented in Sec.~\ref{formaldpl:paioli-asm}.
The details of the code transformation are given in Sec.~\ref{formaldpl:paioli-transform}.
Finally, a proof of the correctness of this transformation is presented in Sec.~\ref{formaldpl:paioli-transformation-proof}.

We implemented \paioli\footref{formaldpl:fn:paioli} using the OCaml\footnote{\url{http://ocaml.org/}} programming language,
which type safety helps to prevent many bugs.
On our \present case-study, it runs in negligible time ($\ll 1$~second), both for \ac{DPL} transformation and simulation, including balance verification.
The unprotected (resp. \ac{DPL}) bitslice AVR assembly file consists of $641$ (resp. $1456$) lines of code.
We use nibble-wise jumps in each \present operation, and an external loop over all rounds.

\subsection{Generic Assembly Language}
\label{formaldpl:paioli-asm}

Our assembly language is generic in that it uses a restricted set of instructions that can be mapped to and from virtually any actual assembly language.
It has the classical features of assembly languages:
logical and arithmetical instructions, branching, labels, direct and indirect addressing.
Fig.~\ref{formaldpl:fig-generic_asm} gives the \ac{BNF} of the language while Fig.~\ref{formaldpl:fig-asm_macro} gives the equivalent code of Fig.~\ref{formaldpl:fig-dpl_macro} as an example of its usage.

\begin{figure}
  \begin{alltt}
  Prog    ::= ( Label? Inst? ( ';' <comment> )? '\string\n' )*
  Label   ::= <label-name> ':'
  Inst    ::= Opcode0
            | Branch1 Addr
            | Opcode2 Lval Val
            | Opcode3 Lval Val Val
            | Branch3 Val Val Addr
  Opcode0 ::= 'nop'
  Branch1 ::= 'jmp'
  Opcode2 ::= 'not' | 'mov'
  Opcode3 ::= 'and' | 'orr' | 'xor' | 'lsl' | 'lsr'
            | 'add' | 'mul'
  Branch3 ::= 'beq' | 'bne'
  Val     ::= Lval | '#' <immediate-value>
  Lval    ::= 'r' <register-number>
            | '@' <memory-address>
            | '!' Val ( ',' <offset> )?
  Addr    ::= '#' <absolute-code-address>
            | <label-name>
\end{alltt}
  \caption{\label{formaldpl:fig-generic_asm} Generic assembly syntax (\ac{BNF}).}
\end{figure}

The semantics of the instructions are intuitive.
For {\tt Opcode2} and {\tt Opcode3} the first operand is the destination and the other are the arguments.
The {\tt mov} instruction is used to copy registers, load a value from memory, or store a value to memory depending on the form of its arguments.
We remark that the instructions use the ``{\tt instr dest op1 op2}'' format, which allows to map similar instructions from $32$-bit processors directly, as well as instructions from $8$-bit processors which only have two operands, by using the same register for {\tt dest} and {\tt op1} for instance.

\subsection{Code Transformation}
\label{formaldpl:paioli-transform}

\paragraph{Bitsliced code.}
As seen in Sec.~\ref{formaldpl:dpl}, \ac{DPL} works at the bit level. Transforming code to make it \ac{DPL} compliant thus requires this level of granularity.
Bitslicing is possible on any algorithm\footnote{Intuitively, the proof invokes the Universal Turing Machines equivalence (those that work with only $\{0,1\}$ as alphabet are as powerful as the others).}, but we found that bitslicing an algorithm is hard to do automatically.
In practice, every bitslice implementations we found were hand-crafted.
However, since Biham presented his bitslice paper~\cite{DBLP:conf/fse/Biham97a}, many block ciphers have been implemented in bitslice for performance reasons, which mitigate this concern.
So, for the sake of simplicity, we assume that the input code is already bitsliced.

\paragraph{\ac{DPL} macros expansion.}
This is the main point of the transformation of the code.

\begin{definition}[Sensitive value]
\label{formaldpl:def-senible_val}
A \emph{value} is said \emph{sensitive} if it depends on sensitive data.
A sensitive data depends on the secret key or the plaintext%
\footnote{Other works consider that a sensitive data must depend on both the secret key and the plaintext (as it is usually admitted in the ``\emph{only computation leaks}'' paradigm; see for instance~\cite[\S 4.1]{DBLP:conf/ches/RivainP10}).
Our definition is broader, in particular it also encompasses the random probing model~\cite{DBLP:conf/crypto/IshaiSW03}.}.
\end{definition}

\begin{definition}[Sensitive instruction]
\label{formaldpl:def-senible_instr}
We say that an \emph{instruction} is \emph{sensitive} if it may modify the Hamming weight of a sensitive value.
\end{definition}

\begin{wrapfigure}{r}{0.25\textwidth}
\begin{alltt}
 mov r1 r0
 mov r1 \textit{a}
 and r1 r1 #3
 lsl r1 r1 #1
 lsl r1 r1 #1
 mov r2 r0
 mov r2 \textit{b}
 and r2 r2 #3
 orr r1 r1 r2
 mov r3 r0
 mov r3 !r1,\textit{op}
 mov \textit{d} r0
 mov \textit{d} r3
\end{alltt}
\caption{\label{formaldpl:fig-asm_macro} \protect\centering \ac{DPL} macro of Fig.~\ref{formaldpl:fig-dpl_macro} in assembly.}
\vspace*{-5mm}
\end{wrapfigure}

All the sensitive instructions must be expanded to a \ac{DPL} macro.
Thus, all the sensitive data must be transformed too.
Each literal (``immediate'' values in assembly terms), memory cells that contain initialized constant data (look-up tables, etc.), and registers values need to be \ac{DPL} encoded.
For instance, using the two least significant bits, the $1$s stay $1$s ({\tt 01}) and the $0$s become $2$s ({\tt 10}).

Since the implementation is bitsliced, only the logical (bit level) operators are used in sensitive instructions ({\tt and}, {\tt or}, {\tt xor}, {\tt lsl}, {\tt lsr}, and {\tt not}).
To respect the \ac{DPL} protocol, {\tt not} instructions are replaced by {\tt xor} which inverse the positive logic and the negative logic bits of \ac{DPL} encoded values.
For instance if using the two least significant bits for the \ac{DPL} encoding, {\tt not $a$ $b$} is replaced by {\tt xor $a$ $b$ \#3}.
Bitsliced code never needs to use shift instructions since all bits are directly accessible.

Moreover, we currently run this code transformation only on block ciphers.
Given that the code is supposed to be bitsliced, this means that the branching and arithmetic instructions are either not used or are used only in a deterministic way (\eg, looping on the round counter) that does not depend on sensitive information.

Thus, only {\tt and}, {\tt or}, and {\tt xor} instructions need to be expanded to \ac{DPL} macros such as the one shown in Fig.~\ref{formaldpl:fig-asm_macro}.
This macro has the advantage that it actually uses two operands instructions only (when there are three operands in our generic assembly language, the destination is the same as one of the two others), which makes its instructions mappable one-to-one even with $8$-bit assembly languages.

\paragraph{Look-up tables.}
As they appear in the \ac{DPL} macro, the addresses of look-up tables are sensitive too.
As seen in Sec.~\ref{formaldpl:dpl-sw}, the look-up tables must be located at an address which is a multiple of $16$ so that the last four bits are
available when adding the offset (in the case where we use the last four bits to place the two \ac{DPL} encoded operands).
Tab.~\ref{formaldpl:tab-dpl_lut} present the $16$ values present in the look-up tables for {\tt and}, {\tt or}, and {\tt xor}.

\begin{table}
  \centering
  \caption{\label{formaldpl:tab-dpl_lut} Look-up tables for {\tt and}, {\tt or}, and {\tt xor}.}
  \begin{scriptsize}
    \begin{tabular}{rl}
      \toprule
      idx & {\tt 0000}, {\tt 0001}, {\tt 0010}, {\tt 0011}, {\tt 0100}, {\tt 0101}, {\tt 0110}, {\tt 0111}, {\tt 1000}, {\tt 1001}, {\tt 1010}, {\tt 1011}, {\tt 1100}, {\tt 1101}, {\tt 1110}, {\tt 1111} \\
      \midrule
      {\tt and} & {\tt ~00~}, {\tt ~00~}, {\tt ~00~}, {\tt ~00~}, {\tt ~00~}, {\tt ~01~}, {\tt ~10~}, {\tt ~00~}, {\tt ~00~}, {\tt ~10~}, {\tt ~10~}, {\tt ~00~}, {\tt ~00~}, {\tt ~00~}, {\tt ~00~}, {\tt ~00~} \\
      {\tt or}  & {\tt ~00~}, {\tt ~00~}, {\tt ~00~}, {\tt ~00~}, {\tt ~00~}, {\tt ~01~}, {\tt ~01~}, {\tt ~00~}, {\tt ~00~}, {\tt ~01~}, {\tt ~10~}, {\tt ~00~}, {\tt ~00~}, {\tt ~00~}, {\tt ~00~}, {\tt ~00~} \\
      {\tt xor} & {\tt ~00~}, {\tt ~00~}, {\tt ~00~}, {\tt ~00~}, {\tt ~00~}, {\tt ~10~}, {\tt ~01~}, {\tt ~00~}, {\tt ~00~}, {\tt ~01~}, {\tt ~10~}, {\tt ~00~}, {\tt ~00~}, {\tt ~00~}, {\tt ~00~}, {\tt ~00~} \\
      \bottomrule
    \end{tabular}
  \end{scriptsize}
\end{table}

Values in the look-up tables which are not at \ac{DPL} valid addresses, \ie,
addresses which are not a concatenation of {\tt 01} or {\tt 10} with {\tt 01}
or {\tt 10}, are preferentially \ac{DPL} invalid, \ie, {\tt 00} or {\tt 11}.
Like this if an error occurs during the execution (such as a fault injection for instance) it poisons the result and all the subsequent computations will be faulted too (infective computation).

\subsection{Correctness Proof of the Transformation}
\label{formaldpl:paioli-transformation-proof}

Formally proving the correctness of the transformation requires to define what we intend by ``correct''.
Intuitively, it means that the transformed code does the ``same thing'' as the original one.

\begin{definition}[Correct \ac{DPL} transformation]
\label{formaldpl:def-correctness}
Let $S$ be a valid state of the system (values in registers and memory).
Let $c$ be a sequence of instructions of the system.
Let $\widehat{S}$ be the state of the system after the execution of $c$ with state $S$, we denote that by $S \xrightarrow{c} \widehat{S}$.
We write $dpl(S)$ for the \ac{DPL} state (with \ac{DPL} encoded values of the $1$s and $0$s in memory and registers) equivalent to the state $S$.\\
We say that $c'$ is a \emph{correct \ac{DPL} transformation} of $c$ if $S \xrightarrow{c} \widehat{S} \implies dpl(S) \xrightarrow{c'} dpl(\widehat{S})$.
\end{definition}

\begin{proposition}[Correctness of our code transformation]
\label{formaldpl:prop-correctness}
The expansion of the sensitive instructions into \ac{DPL} macros such as presented in Sec.~\ref{formaldpl:dpl-sw} is a correct \ac{DPL} transformation.
\end{proposition}

\begin{proof}
Let $a$ and $b$ be instructions.
Let $c$ be the code $a; b$ (instruction $a$ followed by instruction $b$).
Let $X$, $Y$, and $Z$ be states of the program.
If we have $X \xrightarrow{a} Y$ and $Y \xrightarrow{b} Z$,
then we know that $X \xrightarrow{c} Z$ (by \emph{transitivity}).

Let $a'$ and $b'$ be the \ac{DPL} macro expansions of instructions $a$ and $b$.
Let $c'$ be the \ac{DPL} transformation of code $c$.
Since the expansion into macros is done separately for each sensitive instruction,
without any other dependencies, we know that $c'$ is $a'; b'$.\\
If we have $dpl(X) \xrightarrow{a'} dpl(Y)$ and $dpl(Y) \xrightarrow{b'} dpl(Z)$,
then we know that $dpl(X) \xrightarrow{c'} dpl(Z)$.

This means that a chain of correct transformations is a correct transformation.
Thus, we only have to show that the \ac{DPL} macro expansion is a correct transformation.

Let us start with the {\tt and} operation.
Since the code is bitsliced, there are only four possibilities.
Tab.~\ref{formaldpl:proof-and_table} shows these possibilities for the {\tt and d a b} instruction.

Tab.~\ref{formaldpl:proof-and_macro_steps} shows the evolution of the values of $a$, $b$, and $d$ during the execution of the macro which {\tt and d a b} expands to.
We assume the look-up table for {\tt and} is located at address $and$.
Tab.~\ref{formaldpl:proof-and_macro_resume} sums up the Tab.~\ref{formaldpl:proof-and_macro_steps} in the same format as Tab.~\ref{formaldpl:proof-and_table}.

\begin{table}
  \begin{minipage}{0.1\textwidth}
    ~ \\[14mm]
    \begin{scriptsize}
      \begin{tabular}{r}
        Before \\[0.5mm]
        After
      \end{tabular}
    \end{scriptsize}
  \end{minipage}
  \begin{minipage}{0.35\textwidth}
    \caption{\label{formaldpl:proof-and_table} {\tt and\,d\,a\,b}.}
    \begin{scriptsize}
      \begin{tabular*}{0.98\textwidth}{@{\extracolsep{\fill}}cccc}
        \toprule
        \multicolumn{4}{c}{$a, b, d$} \\
        \midrule
        {\tt 0}, {\tt 0}, {\tt ?} & {\tt 0}, {\tt 1}, {\tt ?} & {\tt 1}, {\tt 0}, {\tt ?} & {\tt 1}, {\tt 1}, {\tt ?} \\
        {\tt 0}, {\tt 0}, {\tt 0} & {\tt 0}, {\tt 1}, {\tt 0} & {\tt 1}, {\tt 0}, {\tt 0} & {\tt 1}, {\tt 1}, {\tt 1} \\
        \bottomrule
      \end{tabular*}
    \end{scriptsize}
  \end{minipage}
  \begin{minipage}{0.02\textwidth}~\end{minipage}
  \begin{minipage}{0.48\textwidth}
    \caption{\label{formaldpl:proof-and_macro_resume} \ac{DPL} {\tt and\,d\,a\,b}.}
    \begin{scriptsize}
      \begin{tabular*}{0.98\textwidth}{@{\extracolsep{\fill}}cccc}
        \toprule
        \multicolumn{4}{c}{$a, b, d$} \\
        \midrule
        {\tt 10}, {\tt 10}, {\tt ~?} & {\tt 10}, {\tt 01}, {\tt ~?} & {\tt 01}, {\tt 10}, {\tt ~?} & {\tt 01}, {\tt 01}, {\tt ~?} \\
        {\tt 10}, {\tt 10}, {\tt 10} & {\tt 10}, {\tt 01}, {\tt 10} & {\tt 01}, {\tt 10}, {\tt 10} & {\tt 01}, {\tt 01}, {\tt 01} \\
        \bottomrule
      \end{tabular*}
    \end{scriptsize}
  \end{minipage}
  \vspace*{\baselineskip}
\end{table}

\begin{sidewaystable}
  \centering
  \caption{\label{formaldpl:proof-and_macro_steps} Execution of the \ac{DPL} macro expanded from {\tt and d a b}.}
  \begin{tabular}{lcccclcccclcccclccccl}
    \toprule
    $a, b$ & \multicolumn{4}{c}{{\tt 10}, {\tt 10}} && \multicolumn{4}{c}{{\tt 10}, {\tt 01}} && \multicolumn{4}{c}{{\tt 01}, {\tt 10}} && \multicolumn{4}{c}{{\tt 01}, {\tt 01}} & \\
    & $d$ & {\tt r1} & {\tt r2} & {\tt r3} && $d$ & {\tt r1} & {\tt r2} & {\tt r3} && $d$ & {\tt r1} & {\tt r2} & {\tt r3} && $d$ & {\tt r1} & {\tt r2} & {\tt r3} & \\
    \cmidrule(r){2-6}
    \cmidrule(r){7-11}
    \cmidrule(r){12-16}
    \cmidrule(r){17-21}
    {\tt mov r1 r0}                     & {\tt ~?} & {\tt ~~~0} & {\tt ~~?} & {\tt ~~?} && {\tt ~?} & {\tt ~~~0} & {\tt ~~?} & {\tt ~~?} && {\tt ~?} & {\tt ~~~0} & {\tt ~~?} & {\tt ~~?} && {\tt ~?} & {\tt ~~~0} & {\tt ~~?} & {\tt ~~?} & \\
    {\tt mov r1 $a$}                    & {\tt ~?} & {\tt ~~10} & {\tt ~~?} & {\tt ~~?} && {\tt ~?} & {\tt ~~10} & {\tt ~~?} & {\tt ~~?} && {\tt ~?} & {\tt ~~01} & {\tt ~~?} & {\tt ~~?} && {\tt ~?} & {\tt ~~01} & {\tt ~~?} & {\tt ~~?} & \\
    {\tt and r1 r1 \#3}                 & {\tt ~?} & {\tt ~~10} & {\tt ~~?} & {\tt ~~?} && {\tt ~?} & {\tt ~~10} & {\tt ~~?} & {\tt ~~?} && {\tt ~?} & {\tt ~~01} & {\tt ~~?} & {\tt ~~?} && {\tt ~?} & {\tt ~~01} & {\tt ~~?} & {\tt ~~?} & \\
    {\tt shl r1 r1 \#1}                 & {\tt ~?} & {\tt ~100} & {\tt ~~?} & {\tt ~~?} && {\tt ~?} & {\tt ~100} & {\tt ~~?} & {\tt ~~?} && {\tt ~?} & {\tt ~010} & {\tt ~~?} & {\tt ~~?} && {\tt ~?} & {\tt ~010} & {\tt ~~?} & {\tt ~~?} & \\
    {\tt shl r1 r1 \#1}                 & {\tt ~?} & {\tt 1000} & {\tt ~~?} & {\tt ~~?} && {\tt ~?} & {\tt 1000} & {\tt ~~?} & {\tt ~~?} && {\tt ~?} & {\tt 0100} & {\tt ~~?} & {\tt ~~?} && {\tt ~?} & {\tt 0100} & {\tt ~~?} & {\tt ~~?} & \\
    {\tt mov r2 r0}                     & {\tt ~?} & {\tt 1000} & {\tt ~~0} & {\tt ~~?} && {\tt ~?} & {\tt 1000} & {\tt ~~0} & {\tt ~~?} && {\tt ~?} & {\tt 0100} & {\tt ~~0} & {\tt ~~?} && {\tt ~?} & {\tt 0100} & {\tt ~~0} & {\tt ~~?} & \\
    {\tt mov r2 $b$}                    & {\tt ~?} & {\tt 1000} & {\tt ~10} & {\tt ~~?} && {\tt ~?} & {\tt 1000} & {\tt ~01} & {\tt ~~?} && {\tt ~?} & {\tt 0100} & {\tt ~10} & {\tt ~~?} && {\tt ~?} & {\tt 0100} & {\tt ~01} & {\tt ~~?} & \\
    {\tt and r2 r2 \#3}                 & {\tt ~?} & {\tt 1000} & {\tt ~10} & {\tt ~~?} && {\tt ~?} & {\tt 1000} & {\tt ~01} & {\tt ~~?} && {\tt ~?} & {\tt 0100} & {\tt ~10} & {\tt ~~?} && {\tt ~?} & {\tt 0100} & {\tt ~01} & {\tt ~~?} & \\
    {\tt orr r1 r1 r2}                  & {\tt ~?} & {\tt 1010} & {\tt ~10} & {\tt ~~?} && {\tt ~?} & {\tt 1001} & {\tt ~01} & {\tt ~~?} && {\tt ~?} & {\tt 0110} & {\tt ~10} & {\tt ~~?} && {\tt ~?} & {\tt 0101} & {\tt ~01} & {\tt ~~?} & \\
    {\tt mov r3 r0}                     & {\tt ~?} & {\tt 1010} & {\tt ~10} & {\tt ~~0} && {\tt ~?} & {\tt 1001} & {\tt ~01} & {\tt ~~0} && {\tt ~?} & {\tt 0110} & {\tt ~10} & {\tt ~~0} && {\tt ~?} & {\tt 0101} & {\tt ~01} & {\tt ~~0} & \\
    {\tt mov r3 !r1,$and$}\footnotemark & {\tt ~?} & {\tt 1010} & {\tt ~10} & {\tt ~10} && {\tt ~?} & {\tt 1001} & {\tt ~01} & {\tt ~10} && {\tt ~?} & {\tt 0110} & {\tt ~10} & {\tt ~10} && {\tt ~?} & {\tt 0101} & {\tt ~01} & {\tt ~01} & \\
    {\tt mov $d$ r0}                    & {\tt ~0} & {\tt 1010} & {\tt ~10} & {\tt ~10} && {\tt ~0} & {\tt 1001} & {\tt ~01} & {\tt ~10} && {\tt ~0} & {\tt 0110} & {\tt ~10} & {\tt ~10} && {\tt ~0} & {\tt 0101} & {\tt ~01} & {\tt ~01} & \\
    {\tt mov $d$ r3}                    & {\tt 10} & {\tt 1010} & {\tt ~10} & {\tt ~10} && {\tt 10} & {\tt 1001} & {\tt ~01} & {\tt ~10} && {\tt 10} & {\tt 0110} & {\tt ~10} & {\tt ~10} && {\tt 01} & {\tt 0101} & {\tt ~01} & {\tt ~01} & \\
    \bottomrule
  \end{tabular}
\end{sidewaystable}

This proves that the \ac{DPL} transformation of the {\tt and} instructions are correct.
The demonstration is similar for {\tt or} and {\tt xor} operations.
\end{proof}

The automatic \ac{DPL} transformation of arbitrary assembly code has been implemented in our tool described in App.~\ref{sensi:paioli}.

\section{Formally Proving the Absence of Leakage}
\label{formaldpl:paioli-proof}

Now that we know the \ac{DPL} transformation is correct, we need to prove its efficiency security-wise.
We prove the absence of leakage on the software, while obviously the leakage heavily depends on the hardware.
Our proof thus makes an hypothesis on the hardware: we suppose that the bits we use for the positive and negative logic in the \ac{DPL} protocol leak the same amount.
This may seem like an unreasonable hypothesis, since it is not true in general.
However,
the protection can be implemented in a soft \acs{CPU} core (LatticeMicro32, OpenRISC, LEON2, etc.),
that would be laid out in a \ac{FPGA} or in an \ac{ASIC} with special balancing constraints at place-and-route.
The methodology follows the guidelines given by Chen \etal in~\cite{DBLP:journals/tc/ChenSS13}.
Moreover, we will show in Sec.~\ref{formaldpl:paioli-hw} how it is possible, using stochastic profiling, to find bits which leakages are similar enough for the \ac{DPL} countermeasure to be sufficiently efficient even on non-specialized hardware.
That said, it is important to note that the difference in leakage between two bits of the same register should not be large enough for the attacker to break the \ac{DPL} protection using \ac{SPA} or \ac{ASCA}.

Formally proving the balance of \ac{DPL} code requires to properly define the notions we are using.

\footnotetext{See Tab.~\ref{formaldpl:tab-dpl_lut}.} 

\begin{definition}[Leakage model]
\label{formaldpl:def-leakage_model}
The attacker is able to measure the power consumption of parts of the cryptosystem.
We model power consumption by the Hamming distance of values updates, \ie, the number of bit flips.
It is a commonly accepted model for power analysis, for instance with \ac{DPA}~\cite{kocher-dpa_and_related_attacks} or \ac{CPA}~\cite{cpa-ches04}.
We write $H(a, b)$ the Hamming distance between the values $a$ and $b$.
\end{definition}

\begin{definition}[Constant activity]
\label{formaldpl:def-constant_activity}
The activity of a cryptosystem is said to be constant if its power consumption does not depend on the sensitive data and is thus always the same.\\
Formally, let $P(s)$ be a program which has $s$ as parameter (\eg, the key and the plaintext).
According to our leakage model, a program $P(s)$ is of \emph{constant activity} if:
\begin{compactitem}
\item for every values $s_1$ and $s_2$ of the parameter $s$, for each cycle $i$, for every sensitive value $v$, $v$ is updated at cycle $i$ in the run of $P(s_1)$ if and only if it is updated also at cycle $i$ in the run of $P(s_2)$;
\item whenever an instruction modifies a sensitive value from $v$ to $v'$, then the value of $H(v,v')$ does not depend on $s$.
\end{compactitem}
\end{definition}

\begin{remark}
The first condition of Def.~\ref{formaldpl:def-constant_activity} mostly concerns leakage in the horizontal\,/\,time dimension,
while the second condition mostly concerns leakage in the vertical\,/\,amplitude dimension.
\end{remark}

\begin{remark}
The first condition of Def.~\ref{formaldpl:def-constant_activity} implies that the runs of the program $P(s)$ are constant in time for every $s$.
This implies that a program of constant activity is not vulnerable to timing attacks, which is not so surprising given the similarity between \ac{SPA} and timing attacks.
\end{remark}

\subsection{Computed Proof of Constant Activity}
\label{formaldpl:paioli-cpu}

To statically determine if the code is correctly balanced
(\ie, that the activity of a given program is constant according to Def.~\ref{formaldpl:def-constant_activity}),
our tool relies on symbolic execution.
The idea is to run the code of the program independently of the sensitive data.
This is achieved by computing on sets of all the possible values instead of values directly.
The symbolic execution terminates in our case because we are using the \ac{DPL} protection on block ciphers,
and we avoid combinatorial explosion thanks to bitslicing, as a value can initially be only $1$ or $0$ (or rather their \ac{DPL} encoded counterparts).
Indeed, bitsliced code only use logical instructions as explained in Sec.~\ref{formaldpl:paioli-transform}, which will always return a result in $\{0,1\}$ when given two values in $\{0,1\}$ as arguments.

Our tool implements an interpreter for our generic assembly language which work with sets of values.
The interpreter is equipped to measure all the possible Hamming distances of each value update, and all the possible Hamming weight of values.
It watches updates in registers, in memory, and also in address buses (since the addresses may leak information when reading in look-up tables).
If for one of these value updates there are different possible Hamming distances or Hamming weight, then we consider that there is a leak of information: the power consumption activity is not constant according to Def.~\ref{formaldpl:def-constant_activity}.

\paragraph{Example.}
Let $a$ be a register which can initially be either $0$ or $1$.
Let $b$ be a register which can initially be only $1$.
The execution of the instruction {\tt orr $a$ $a$ $b$} will set the value of $a$ to be all the possible results of $a \lor b$.
In this example, the new set of possible values of $a$ will be the singleton $\{1\}$ (since $0 \lor 1$ is $1$ and $1 \lor 1$ is $1$ too).
The execution of this instruction only modified one value, that of $a$.
However, the Hamming distance between the previous value of $a$ and its new value can be either $0$ (in case $a$ was originally $1$) or $1$ (in case $a$ was originally $0$).
Thus, we consider that there is a leak.
\smallskip

By running our interpreter on assembly code, we can statically determine if there are leakages or if the code is perfectly balanced.
For instance for a block cipher, we initially set the key and the plaintext (\ie, the sensitive data) to have all their possible values:
all the memory cells containing the bits of the key and of the plaintext have the value $\{0, 1\}$ (which denotes the set of two elements: $0$ and $1$).
Then the interpreter runs the code and outputs all possible leakage; if none are present, it means that the code is well balanced.
Otherwise we know which instructions caused the leak, which is helpful for debugging, and also to locate sensitive portions of the code.

For an example in which the code is balanced, we can refer to the execution of the {\tt and} \ac{DPL} macro shown in Tab.~\ref{formaldpl:proof-and_macro_steps}.
There we can see that the Hamming distance of the updates does not depend on the values of $a$ and $b$.
We also note that at the end of the execution (and actually, all along the execution) the Hamming weight of each value does not depend on $a$ and $b$ either.
This allows to chain macros safely: each value is precharged with $0$ before being written to.

\subsection{Hardware Characterization}
\label{formaldpl:paioli-hw}

The \ac{DPL} countermeasure relies on the fact that the pair of bits used to store the \ac{DPL} encoded values leak the same way, \ie, that their power consumptions are the same.
This property is generally not true in non-specialized hardware.
However, using the two closest bits (in terms of leakage) for the \ac{DPL} protocol still helps reaching a better immunity to side-channel attacks, especially \acp{ASCA} that operate on a limited number of traces.

The idea is to compute the leakage level of each of the bits during the execution of the algorithm,
in order to choose the two closest ones as the pair to use for the \ac{DPL} protocol and thus ensure an optimal balance of the leakage.
This is facilitated by the fact that the algorithm is bitsliced.
Indeed, it allows to run the whole computation using only a chosen bit while all the others stay zero.
We will see in Sec.~\ref{formaldpl:avr-profiling} how we characterized our smartcard in practice.

\section{Case Study: \present on an ATmega163 AVR Micro-Controller}
\label{formaldpl:avr}

\subsection{Profiling the ATmega163}
\label{formaldpl:avr-profiling}

We want to limit the size of the look-up tables used by the \ac{DPL} macros.
Thus, \ac{DPL} macros need to be able to store two \ac{DPL} encoded bits in the four consecutive bits of a register.
This lets $13$ possible \ac{DPL} encoding layouts on $8$-bit.
Writing {\tt X} for a bit that is used and {\tt x} otherwise, we have:
\begin{compactenum}
\item {\tt xxxxxxXX},
\item {\tt xxxxxXXx},
\item {\tt xxxxXXxx},
\item {\tt xxxXXxxx},
\item {\tt xxXXxxxx},
\item {\tt xXXxxxxx},
\item {\tt XXxxxxxx},
\item {\tt xxxxxXxX},
\item {\tt xxxxXxXx},
\item {\tt xxxXxXxx},
\item {\tt xxXxXxxx},
\item {\tt xXxXxxxx},
\item {\tt XxXxxxxx}.
\end{compactenum}
As explained in Sec.~\ref{formaldpl:paioli-hw}, we want to use the pair of bits that have the closest leakage properties, and also which is the closest from the least significant bit, in order to limit the size of the look-up tables.

To profile the AVR chip (we are working with an \emph{Atmel ATmega163 AVR} smartcard, which is \emph{notoriously leaky}),
we ran eight versions of an unprotected bitsliced implementation of \present, each of them using only one of the $8$ possible bits.
We used the \ac{NICV}~\cite{emc-tokyo}, also called \emph{coefficient of determination},
as a metric to evaluate the leakage level of the variables of each of the $8$ versions.
Let us denote by $L$ the (noisy and non-injective) leakage associated with the manipulation of the sensitive value $V$,
both seen as random variables;
then the \ac{NICV} is defined as the ratio between the inter-class and the total variance of the leakage,
that is:
$
\text{\ac{NICV}} = \frac{\mathsf{Var}\lbrack\mathbb{E}\lbrack L|V \rbrack\rbrack}{\mathsf{Var}\lbrack L \rbrack} \enspace.
$
By the Cauchy-Schwarz theorem, we have $0\leqslant \text{\ac{NICV}}\leqslant 1$; thus the \ac{NICV} is an absolute leakage metric.
A key advantage of \ac{NICV} is that it detects leakage using public information like input plaintexts or output ciphertexts only.
We used a fixed key and a variable plaintext on which applying \ac{NICV} gave us the leakage level of all the intermediate variables in bijective relation with the plaintext
(which are all the sensible data as seen in Def.~\ref{formaldpl:def-senible_val}).
As we can see on the measures plotted in Fig.~\ref{formaldpl:fig-profiling_bits} (which can be found in App.~\ref{formaldpl:app-profiling}),
the least significant bit leaks very differently from the others,
which are roughly equivalent in terms of leakage\footnote{These differences are due to the internal architecture of the chip, for which we don't have the specifications.}.
Thus, we chose to use the {\tt xxxxxXXx} \ac{DPL} pattern to avoid the least significant bit (our goal here is not to use the optimal pair of bits but rather to demonstrate the added-value of the characterization).

\subsection{Generating Balanced AVR Assembly}
\label{formaldpl:avr-mapping}

We wrote an AVR bitsliced implementation of \present that uses the S-Box in $14$ logic gates from Courtois \etal~\cite{DBLP:journals/iacr/CourtoisHM11}.
This implementation was translated in our generic assembly language (see Sec.~\ref{formaldpl:paioli-asm}).
The resulting code was balanced following the method discussed in Sec.~\ref{formaldpl:paioli-gen}, except that we used the \ac{DPL} encoding layout adapted to our particular smartcard, as explained in Sec.~\ref{formaldpl:avr-profiling}.
App.~\ref{formaldpl:app-macro} presents the code of the adapted \ac{DPL} macro.
The balance of the \ac{DPL} code was then verified as in Sec.~\ref{formaldpl:paioli-proof}.
Finally, the verified code was mapped back to AVR assembly.
All the code transformations and the verification were done automatically using our tool.

\subsection{Cost of the Countermeasure}
\label{formaldpl:avr-cost}


\begin{wraptable}{R}{0.45\textwidth}
  \centering
  \caption{\label{formaldpl:tab-cost} \ac{DPL} cost.}
  \begin{footnotesize}
  \begin{tabular*}{\linewidth}{@{\extracolsep{\fill}}rrrr}
    \toprule
    & bitslice & \ac{DPL} & cost \\
    \midrule
    code (B) & $1620$   & $3056$    & $\times 1.88$ \\
    RAM (B)  & $288$    & $352$     & $+ 64$        \\
    \#cycles & $78,403$ & $235,427$ & $\times 3$    \\
    \bottomrule
  \end{tabular*}
  \end{footnotesize}
\end{wraptable}

The table in Tab.~\ref{formaldpl:tab-cost} compares the performances of the \ac{DPL} protected implementation of \present with the original bitsliced version from which the protected one has been derived.
The \ac{DPL} countermeasure multiplies by $1.88$ the size of the compiled code.
This low factor can be explained by the numerous instructions which it is not necessary to transform (the whole permutation layer of the \present algorithm is left as is for instance).
The protected version uses $64$ more bytes of memory (sparsely, for the \ac{DPL} macro look-up tables).
It is also only $3$ times slower%
\footnote{Notice that \present is inherently slow in software
(optimized non-bitsliced assembly is reported to run in about $11,000$ clock cycles on an \emph{Atmel ATtiny 45} device~\cite{DBLP:conf/africacrypt/EisenbarthGGHIKKNPRSO12})
because it is designed for hardware.
Typically, the permutation layer is free in hardware, but requires many bit-level manipulations in software.
Nonetheless, we precise that there are contexts where \present must be supported,
but no hardware accelerator is available.},
or $24$ times if we consider that the original bitsliced but unprotected code could operate on $8$ blocks at a time.

Note that these experimental results are only valid for the \present algorithm on the \emph{Atmel ATmega163 AVR} device we used.
Further work is necessary to compare these results to those which would be obtained with other algorithms such as \ac{AES}, and on other platforms such as ARM processors.

\subsection{Attacks}
\label{formaldpl:avr-attacks}

We attacked three implementations of the \present algorithm: a bitsliced but unprotected one, a \ac{DPL} one using the two less significant bits, and a \ac{DPL} one using two bits that are more balanced in term of leakage (as explained in Sec.~\ref{formaldpl:avr-profiling}).
On each of these, we computed the success rate of using monobit \ac{CPA} of the output of the S-Box as a model.
The monobit model is relevant because only one bit of sensitive data is manipulated at each cycle since the algorithm is bitsliced,
and also because each register is precharged at $0$ before a new intermediate value is written to it, as per the \ac{DPL} protocol prescribe.
Note that this means we consider the resistance against first-order attacks only.
Actually, we are precisely in the context of~\cite{ifs:mangard}, where the efficiency of correlation and Bayesian attacks gets close as soon as the number of queries required to perform a successful attack is large enough.
This justifies our choice of the \ac{CPA} for the attack evaluation.

The results are shown in Fig.~\ref{formaldpl:fig_attacks} (which can be found in App.~\ref{formaldpl:app-attacks}).
They demonstrate that the first \ac{DPL} implementation is at least $10$ times more resistant to first-order power analysis attacks (requiring almost $1,500$ traces) than the unprotected one.
The second \ac{DPL} implementation, which takes the chip characterization into account, is \emph{$34$ times more resistant} (requiring more than $4,800$ traces).

Interpreting these results requires to bear in mind that the \emph{attacks setting was largely to the advantage of the attacker}.
In fact, these results are very pessimistic: we used our knowledge of the key to select a narrow part of the traces where we knew that the attack would work, and we used the \ac{NICV}~\cite{emc-tokyo} to select the point where the \ac{SNR} of the \ac{CPA} attack is the highest
(see similar use cases of \ac{NICV} in~\cite{Bhasin:2014:SLT:2611765.2611772}).
We did this so we could show the improvement in security due to the characterization of the hardware.
Indeed, without this ``cheating attacker'' (for the lack of a better term), \ie, when we use a monobit \ac{CPA} taking into account the maximum of correlation over the full round, as a normal attacker would do,
the unprotected implementation breaks using about $400$ traces (resp. $138$ for the ``cheating attacker''), while the poorly balanced one is still not broken using $100,000$ traces (resp. about $1,500$).
We do not have more traces than that so we can only say that with an experimental \ac{SNR} of $15$ (which is quite large so far), the security gain is more than $250\times$ and may be much higher with the hardware characterization taken into account as our results with the ``cheating attacker'' shows.

As a comparison%
\footnote{We insist that the comparison between two security gains is very platform-dependent.
The figures we give are only valid on our specific setup. Of course, for different conditions, \eg, lower signal-to-noise ratio, masking might become more secure than \ac{DPL}.},
an unprotected \ac{AES} on the same smartcard breaks in $15$ traces,
and in $336$ traces with a first order masking scheme using less powerful attack setting (see success rates of masking in App.~\ref{formaldpl:sub-masked_aes}),
hence a security gain of $22\times$.
Besides, we notice that our software \ac{DPL} protection thwarts \acp{ASCA}.
Indeed, \acp{ASCA} require a high signal to noise ratio on a single trace.
This can happen both on unprotected and on masked implementation.
However, our protection aims at theoretically cancelling the leakage,
and practically manages to reduce it significantly, even when the chosen \ac{DPL} bit pair is not optimal.
Therefore, coupling software \ac{DPL} with key-update~\cite{MM:AFRICACRYPT10}
allows to both prevent against fast attacks on few traces (\acp{ASCA}) and against attacks that would require more traces (regular \acp{CPA}).

\section{Conclusions and Perspectives}
\label{formaldpl:concl}

\paragraph{Contributions.}

We present a method to protect any bitsliced assembly code by transforming it to enforce the \acf{DPL} protocol, which is a balancing countermeasure against power analysis.
We provide a tool which automates this transformation.
We also formally prove that this transformation is correct, \ie, that it preserves the semantic of the program.

Independently, we show how to formally prove that assembly code is well balanced.
Our tool is also able to use this technique to statically determine whether some arbitrary assembly code's power consumption activity is constant, \ie, that it does not depend on the sensitive data.
In this chapter we used the Hamming weight of values and the Hamming distance of values update as leakage models for power consumption, but our method is not tied to it and could work with any other leakage models that are computable.
We present how to characterize the targeted hardware to make use of the resources which maximize the relevancy of our leakage model to run the \ac{DPL} protocol.

We then applied our methods using our tool using an implementation of the \present cipher on a real smartcard, which ensured that our methods and models are relevant in practice.
In our case study, the provably balanced \ac{DPL} protected implementation is at least $250$ times more resistant to power analysis attacks than the unprotected version while being only $3$ times slower.
These figures could be better.
Indeed, they do not take into account hardware characterization which helps the balancing a lot, as we were able to see with the ``cheating attacker''.
Moreover, we have used the hardware characterization data grossly, only to show the added-value of the operation, which as expected is non-negligible.
And of course interpreting our figures require to take into account that the \emph{ATmega163}, the model of smartcard that we had at our disposal, is notoriously leaky.

These results show that software balancing countermeasures are realistic:
our formally proved countermeasure is an order of magnitude less costly than the state of the art of formally proved masking~\cite{DBLP:conf/ches/RivainP10}.

\paragraph{Future work.}
The first and foremost future work surely is that our methods and tools need to be further tested in other experimental settings, across more hardware platforms, and using other cryptographic algorithms.

We did not try to optimize our \present implementation (neither for speed nor space).
However, automated proofs enable optimization:
indeed, the security properties can be checked again after any optimization attempt (using proofs computation as non-regression tests, either for changes in the \ac{DPL} transformation method, or for handcrafted optimizations of the generated \ac{DPL} code).

Although the mapping from the internal assembly of our tool to the concrete assembly is straightforward, it would be better to have a formal correctness proof of the mapping.

Our work would also benefit from automated bitslicing, which would allow to automatically protect any assembly code with the \ac{DPL} countermeasure.
However, it is still a challenging issue.

Finally, the \ac{DPL} countermeasure itself could be improved:
the pair of bits used for the \ac{DPL} protocol could change during the execution,
or more simply it could be chosen at random for each execution in order to better balance the leakage among multiple traces.
Besides, unused bits could be randomized instead of being zero in order to add noise on top of balancing, and thus reinforce the hypotheses we make on the hardware.
An anonymous reviewer of the PROOFS 2014 workshop suggested that randomness could instead be used to mask the intermediate bits.
Indeed, the reviewer thinks that switching bus lines may only increase noise, while masking variables may provide sound resistance, at least at first order.
The resulting method would therefore:
\begin{inparaenum}
\item gain both the 1st-order resistance of masking countermeasures and the significant flexibility of software-defined countermeasures;
\item still benefit from the increase of resistance resorting to the use of the \ac{DPL} technique, as demonstrated by present chapter.
\end{inparaenum}
This suggestion is of course only intuitive and lacks argumentation based on precise analysis and calculation.
\\[1mm]
\indent\emph{We believe formal methods have a bright future concerning the certification of side-channel attacks countermeasures (including their implementation in assembly) for trustable cryptosystems.}

\bibliographystyle{alpha}
\bibliography{./dpl}

\clearpage
\appendix

\section{\paioli}
\label{sensi:paioli}

The goal of \paioli\footnote{\url{http://pablo.rauzy.name/sensi/paioli.html}} (\emph{Power Analysis Immunity by Offsetting Leakage Intensity}) is to protect assembly code against power analysis attacks such as DPA (differential power analysis) and CPA (correlation power analysis), and to formally prove the efficiency of the protection.
To this end, it implements the automatic insertion of a balancing countermeasure, namely DPL (dual-rail with precharge logic), in assembly code (for now limited to bitsliced block-cipher type of algorithms).
Independently, it is able to statically verify if the power consumption of a given assembly code is correctly balanced with regard to a leakage model (\eg, the Hamming weight of values, or the Hamming distance of values updates).

\begin{footnotesize}
\begin{alltt}
paioli [options] <input-file>
    -bf Bit to use as F is DPL macros (default: 1)
    -bt Bit to use as T is DPL macros (default: 0)
    -po Less significant bit of the DPL pattern for DPL LUT access
        (default: 0)
    -cl Compact the DPL look-up table (LUT) if present
    -la Address in memory where to put the DPL LUT (default: 0)
    -r1 Register number of one of the three used by DPL macros
        (default: 20)
    -r2 Register number of one of the three used by DPL macros
        (default: 21)
    -r3 Register number of one of the three used by DPL macros
        (default: 22)
    -a  Adapter for custom assembly language
    -o  asm output (default: no output)
    -l  Only check syntax if present
    -d  Perform DPL transformation of the code if present
    -v  Perform leakage verification if present
    -s  Perform simulation if present
    -r  Register count for simulation (default: 32)
    -m  Memory size for simulation (default: 1024)
    -M  range of memory to display after simulation
    -R  range of registers to display after simulation
\end{alltt}
\end{footnotesize}

The rest of this section details its features.

\paragraph{Adapters.}
To easily adapt it to any assembly language, it has a system of plugins (which we call ``adapters'') that allows to easily write a parser and a pretty-printer for any language and to use them instead of the internal parser and pretty-printer (which are made for the internal language we use, see Sec.~\ref{formaldpl:paioli-asm}) without having to recompile the whole tool.

\paragraph{DPL transformation.}
If asked so, \paioli is able to automatically apply the DPL transformation as explained in Sec.~\ref{formaldpl:paioli-transform}.
It takes as arguments which bits to use for the DPL protocol,
the offset at which to place the pattern for look-up tables
(for example, we used an offset of $1$ to avoid resorting to the least significant bit which leaks differently),
and where in memory should the look-up tables start.
Given these parameters, the tool verifies that they are valid and consistent according to the DPL protocol, and then it generates the DPL balanced code corresponding to the input code, including the code for look-up tables initialization.
Optionally, the tool is able to compact the look-up tables (since they are sparse), still making sure that their addresses respect the DPL protocol (Sec.~\ref{formaldpl:dpl-sw}).

\paragraph{Simulation.}
If asked so, \paioli can simulate the execution of the code after its optional DPL transformation.
The simulator is equipped to do the balance verification proof (see Sec.~\ref{formaldpl:paioli-proof}) but it is not mandatory to do the balance analysis when running it.
It takes as parameters the size of the memory and the number of register to use, and initializes them to the set of two DPL encoded values of $1$ and $0$ corresponding to the given DPL parameters.
The tool can optionally display the content of selected portions of the memory or of chosen registers after execution, which is useful for inspection and debugging purpose for example.

\paragraph{Balance verification.}
The formal verification of the balance of the code is an essential functionality of the tool.
Indeed, bugs occur even when having a thorough and comprehensive specification,
thus we believe that it is not sufficient to have a precise and formally proven method for generating protected code, but that the results should be independently verified (see Sec.~\ref{formaldpl:paioli-proof}).

\section{Characterization of the Atmel ATmega163 AVR Micro-Controller}
\label{formaldpl:app-profiling}

Fig.~\ref{formaldpl:fig-profiling_bits} shows the leakage level computed using \ac{NICV}~\cite{emc-tokyo} for each bit of the \emph{Atmel ATmega163 AVR} smartcard that we used for our tests (see Sec.~\ref{formaldpl:avr-profiling}).
We can see the first bit leaks very differently from the others.
Thus it is not a good candidate to appear in the bit pair used for the \ac{DPL} protocol.

\begin{figure}[h]
\includegraphics[height=\textwidth,angle=270]{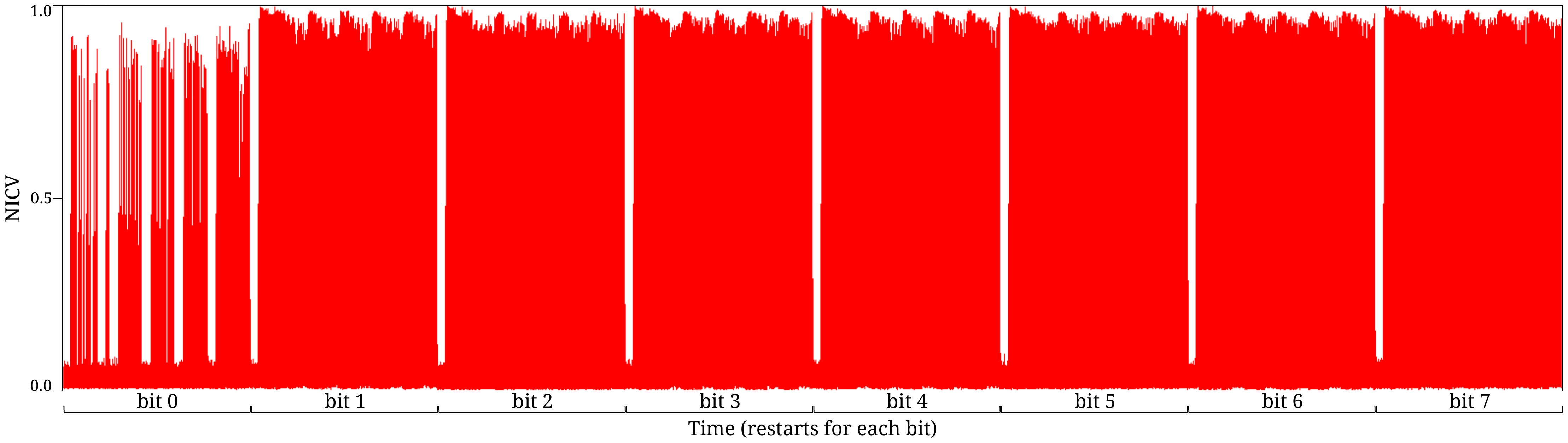}
\caption{Leakage during unprotected encryption for each bit on \emph{ATmega163}.
\label{formaldpl:fig-profiling_bits}}
\end{figure}

\section{\acs{DPL} Macro for the AVR Micro-Controller}
\label{formaldpl:app-macro}

\begin{wrapfigure}{R}{0.28\textwidth}
  \centering
  \begin{tabular}{r c l}
    $r_1$ & $\leftarrow$ & $r_0$ \\
    $r_1$ & $\leftarrow$ & $a$ \\
    $r_1$ & $\leftarrow$ & $r_1 \wedge 6$ \\
    $r_1$ & $\leftarrow$ & $r_1 \ll 1$ \\
    $r_1$ & $\leftarrow$ & $r_1 \ll 1$ \\
    $r_2$ & $\leftarrow$ & $r_0$ \\
    $r_2$ & $\leftarrow$ & $b$ \\
    $r_2$ & $\leftarrow$ & $r_2 \wedge 6$ \\
    $r_1$ & $\leftarrow$ & $r_1 \vee r_2$ \\
    $r_3$ & $\leftarrow$ & $r_0$ \\
    $r_3$ & $\leftarrow$ & $op[r_1]$ \\
    $d$   & $\leftarrow$ & $r_0$ \\
    $d$   & $\leftarrow$ & $r_3$ \\
  \end{tabular}
  \caption{\label{formaldpl:fig-dpl_macro_avr} \protect\centering \ac{DPL} macro for $d = a~\mathrm{op}~b$ on the \emph{ATmega163}.}
  \vspace*{-1.5cm}
\end{wrapfigure}

Once we profiled our smartcard as described in Sec.~\ref{formaldpl:avr-profiling},
we decided to use the bits $1$ and $2$ for the \ac{DPL} protocol ({\tt xxxxxXXx}),
that is, the \ac{DPL} value of $1$ becomes $2$ and the \ac{DPL} value of $0$ becomes $4$.
To avoid using the least significant bit (which leaks very differently from the others),
we decided to align the two \ac{DPL} bits for look-up table access starting on the bit $1$ rather than $0$ ({\tt xxxXXXXx}).
With these settings, the \ac{DPL} macro automatically generated by \paioli is presented in Fig.~\ref{formaldpl:fig-dpl_macro_avr} (it follows the same conventions as Fig.~\ref{formaldpl:fig-dpl_macro}).
As we can see the only modification is the mask applied in the logical {\it and} instructions which is now $6$ instead of $3$ to reflect the new \ac{DPL} pattern.

Note that the least significant bit is now unused by the \ac{DPL} protocol and allowed \paioli to compact the look-up tables used by the \ac{DPL} macros.
Indeed, their addresses need to be of the form {\tt /.+0000./} leaving the least significant bit free and thus allowing to interleave two look-up tables one on another without overlapping of their actually used cells (see Sec.~\ref{formaldpl:paioli-transform}).

\section{Attacks}

\subsection{Attack results on masking (\acs{AES})}
\label{formaldpl:sub-masked_aes}

For the sake of comparison, we provide attack results on the same smartcard tested with the same setup.
Figure~\ref{formaldpl:fig_higher_order} shows the success rate for the attack on the first byte of an \ac{AES}.

We estimate the number of traces for a successful attack as the abscissa where the success rate curve first intersects the $80\%$ horizontal line.

\begin{figure}
  \centering
  \begin{subfigure}[t]{0.48\textwidth}
    \centering
    \caption{{\small Univariate \ac{CPA} attack on unprotected \ac{AES}.}}
    \includegraphics[width=\textwidth,trim=0mm 82.2mm 0mm 84.3mm,clip=true]{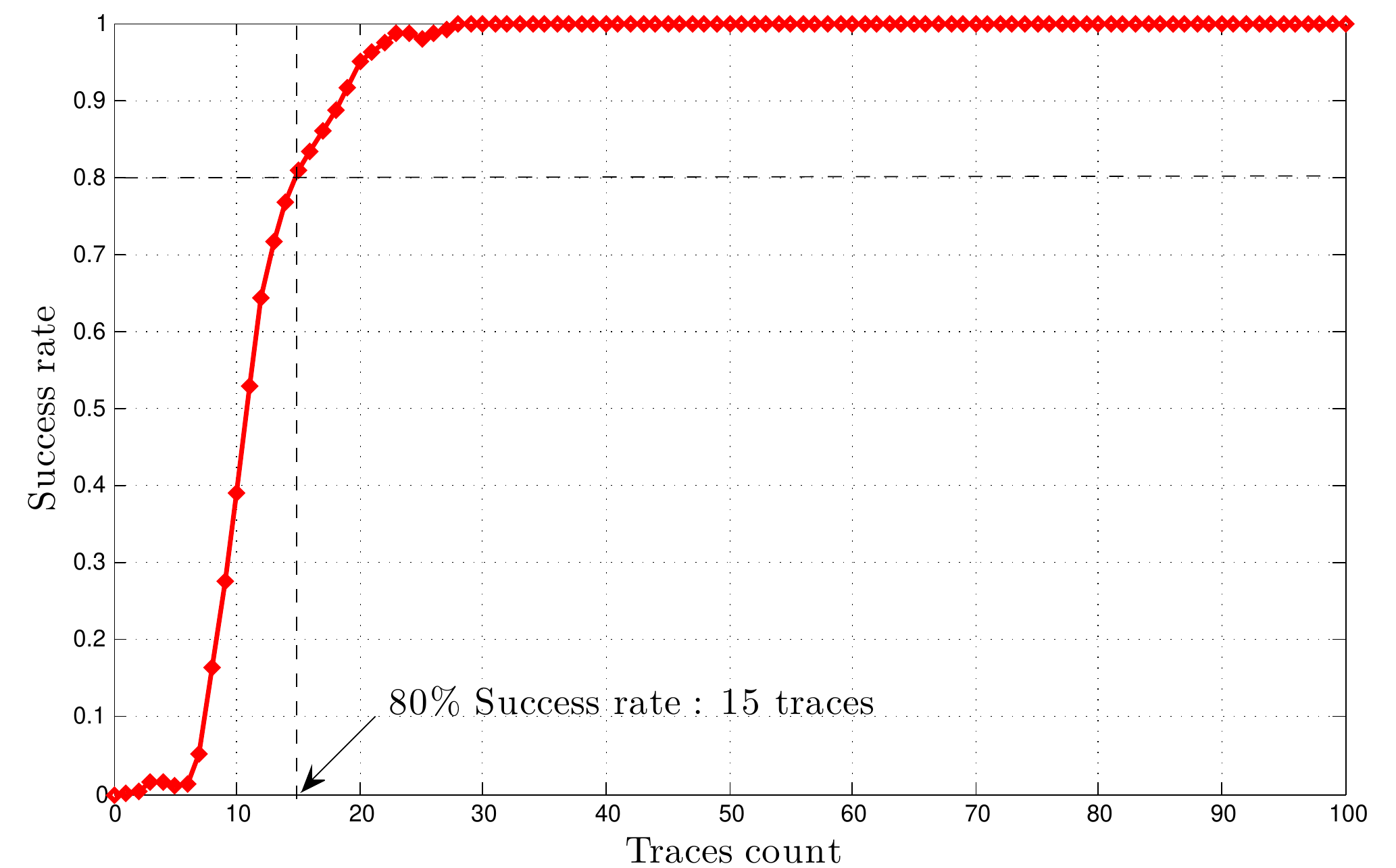}
  \end{subfigure}
  ~
  \begin{subfigure}[t]{0.48\textwidth}
    \centering
    \caption{{\small Bi-variate 2O-\ac{CPA} on 1st-order protected \ac{AES}.}}
    \includegraphics[width=\textwidth,trim=0mm 82mm 0mm 84.8mm,clip=true]{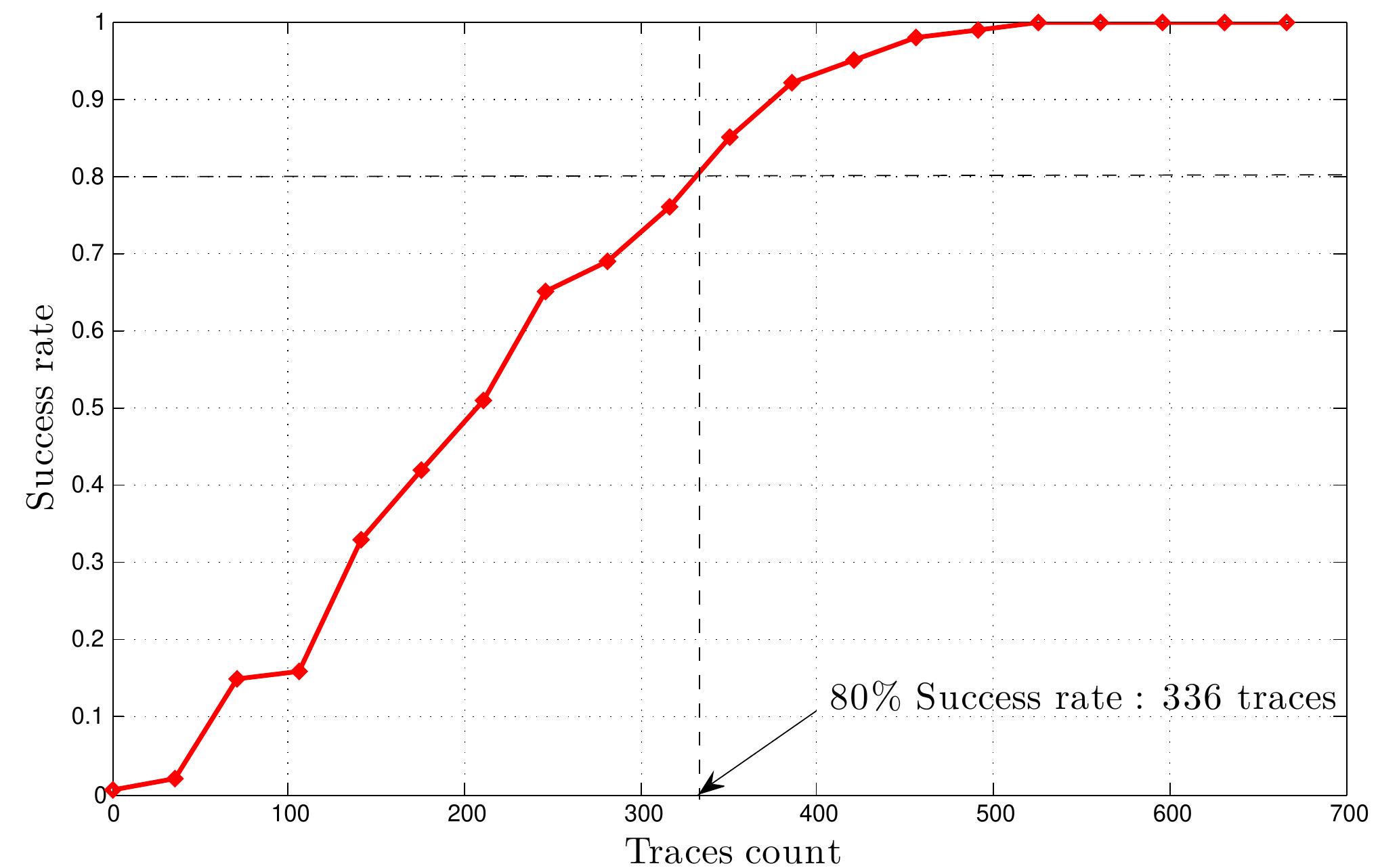}
  \end{subfigure}
  \caption{\label{formaldpl:fig_higher_order} \protect\centering Attacking \ac{AES} on the \emph{ATmega163}: success rates.}
\end{figure}

\subsection{Attack results on \acs{DPL} (\present)}
\label{formaldpl:app-attacks}

Fig.~\ref{formaldpl:fig_attacks} shows the success rates and the correlation curves when attacking our three implementations of \present.
The sensitive variable we consider is in line with the choice of Kocher \etal in their CRYPTO'99 paper~\cite{kocher-dpa_and_related_attacks}:
it is the least significant bit of the output of the substitution boxes (that are $4 \times 4$ in \present).

In Fig.~\ref{formaldpl:fig-timings}, we give, for the unprotected bitslice implementation,
the correspondence between the operations of \present and the \ac{NICV} trace.
The zones of largest \ac{NICV} correspond to operations that access (read or write) sensitive data in RAM.
To make the attacks more powerful, they are not done on the maximal correlation point over the full first round of \present%
\footnote{Note that using the maximum correlation point to attack the \ac{DPL} implementations resulted in the success rate remaining always at $\approx 1/{16}$ (there are $2^4$ key guesses in \present when targeting the first round, because the substitution boxes are $4 \times 4$) in average (at least on the number of traces we had ($100,000$)) on both on them.}
($500,000$ samples),
but rather on a smaller interval
(of only $140$ samples, \ie, one clock period of the device)
of high potential leakage revealed by the \ac{NICV} computations, namely sBoxLayer.

This makes the attack much more powerful and has to be taken into account when interpreting its results.
In fact, the results we present are very pessimistic: we used our knowledge of the key to select a narrow part of the traces where we knew that the attack would work, and we used the \ac{NICV}~\cite{emc-tokyo} to select the point where the \ac{SNR} of the \ac{CPA} attack is the highest.
We did this so we could show the improvement in security due to the characterization of the hardware.
Indeed, without this ``cheating attacker'' (for the lack of a better term), \ie, when we use a monobit \ac{CPA} taking into account the maximum of correlation over the full round, as a normal attacker would do,
the unprotected implementation breaks using about $400$ traces (resp. $138$ for the ``cheating attacker''), while the poorly balanced one is still not broken using $100,000$ traces (resp. about $1,500$).
We do not have more traces than that so we can only say that with an experimental \ac{SNR} of 15 (which is quite large so far), the security gain is more than $250\times$ and may be much higher with the hardware characterization taken into account as our results with the ``cheating attacker'' shows.
Another way of understanding the $250$-fold data complexity increase for the \ac{CPA} is to turn this figure into a reduction of the \ac{SNR}:
according to~\cite{DBLP:conf/ches/ThillardPR13,Bhasin:2014:SLT:2611765.2611772},
our \ac{DPL} countermeasure has attenuated the \ac{SNR} by a factor of at least $\sqrt{250} \approx 16$.

\begin{figure}
  \centering
  \includegraphics[width=\textwidth]{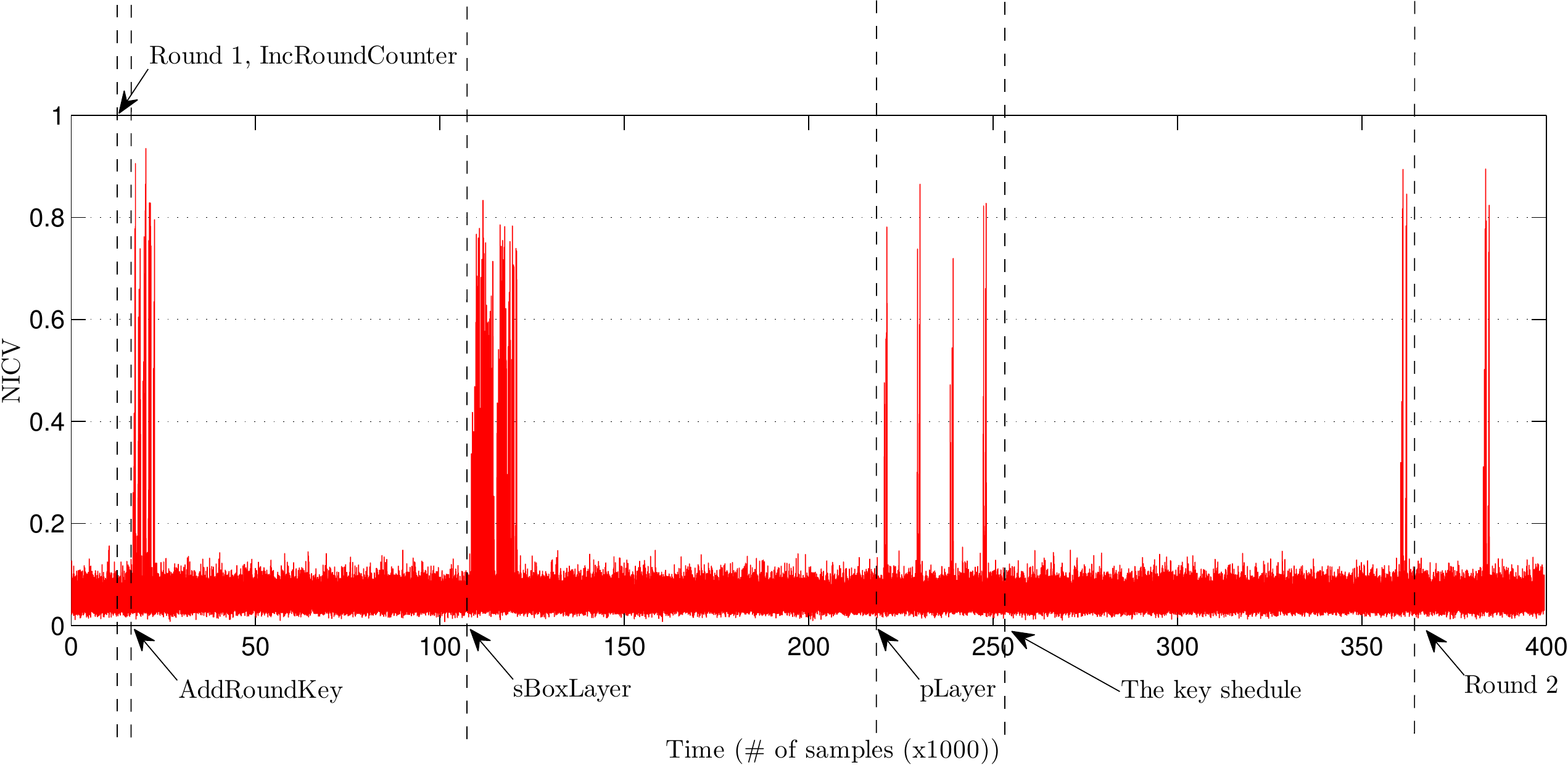}
  \caption{Correspondence between \ac{NICV} and the instructions of \present.}
  \label{formaldpl:fig-timings}
\end{figure}

\begin{figure}
  \centering
  \begin{subfigure}{\textwidth}
    \centering
    \caption{{\small Monobit \ac{CPA} attack on unprotected bitslice implementation.}}
    \includegraphics[height=\textwidth,angle=270]{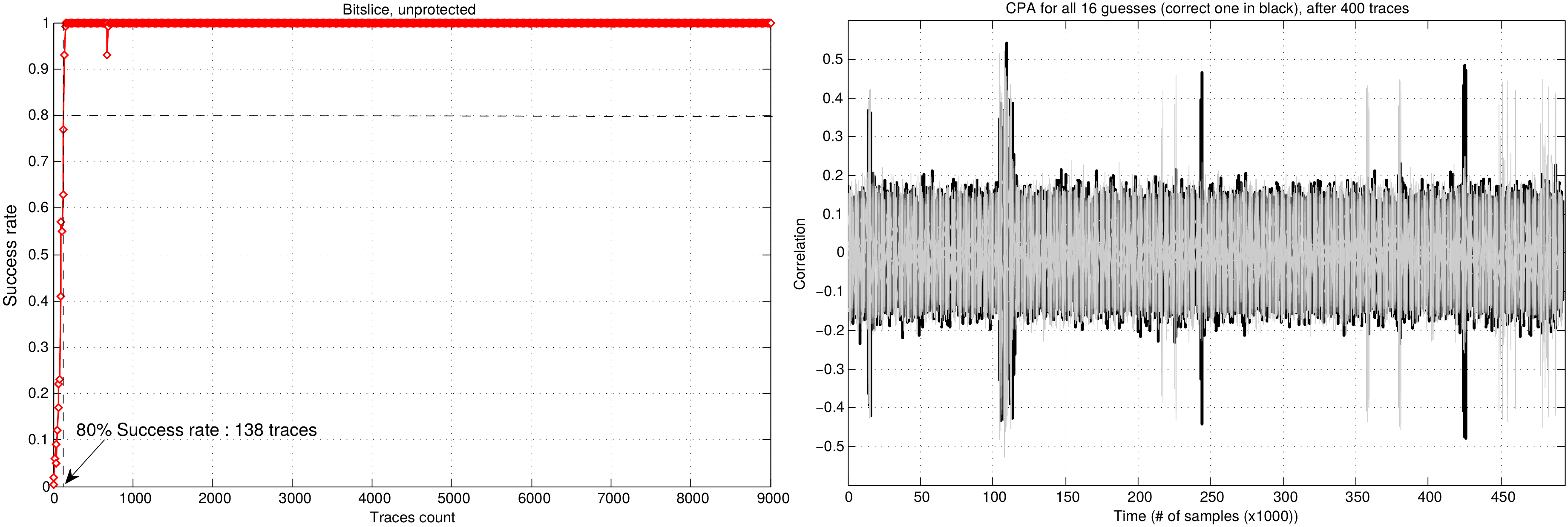}
    \vspace*{\baselineskip}
  \end{subfigure}
  \begin{subfigure}{\textwidth}
    \centering
    \caption{{\small Monobit \ac{CPA} attack on poorly balanced \ac{DPL} implementation (bits $0$ and $1$).}}
    \includegraphics[height=\textwidth,angle=270]{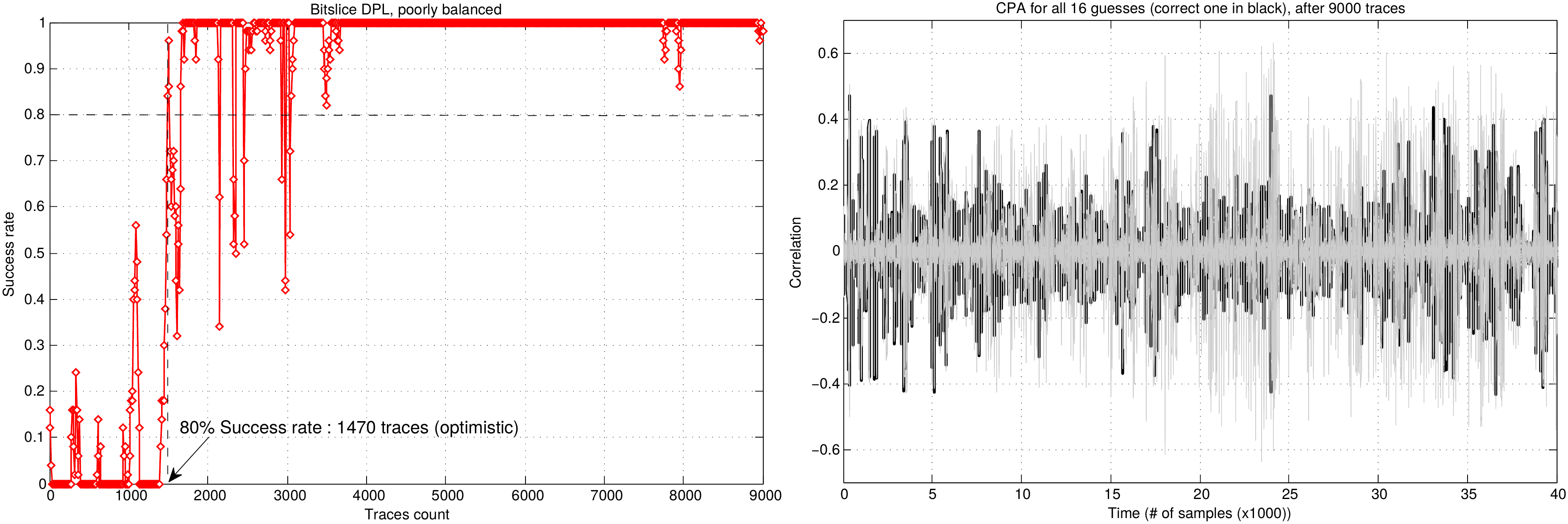}
    \vspace*{\baselineskip}
  \end{subfigure}
  \begin{subfigure}{\textwidth}
    \centering
    \caption{{\small Monobit \ac{CPA} attack on better balanced \ac{DPL} implementation (bits $1$ and $2$).}}
    \includegraphics[height=\textwidth,angle=270]{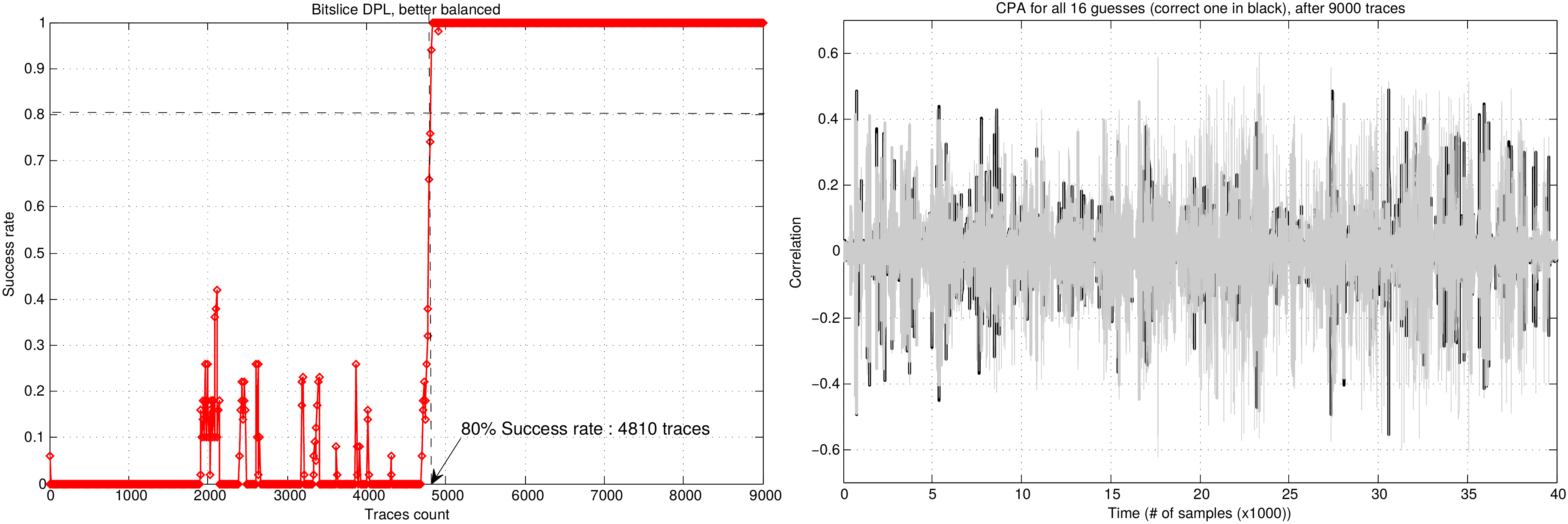}
  \end{subfigure}
  \caption{\label{formaldpl:fig_attacks} \protect\centering Attacks on our three implementations of \present;
    \protect\newline \underline{\emph{Left}}: success rates (estimated with $100$~attacks/step), and
    \protect\newline \underline{\emph{Right}}: \ac{CPA} curves (whole first round in (a), and only sBoxLayer for (b) and (c)).}
\end{figure}

\acrodef{AES}{Advanced Encryption Standard}
\acrodef{ASCA}{Algebraic Side-Channel Attack}
\acrodef{ASIC}{Application-Specific Integrated Circuit}
\acrodef{BCDL}{Balanced Cell-based Differential Logic}
\acrodef{BNF}{Backus--Naur Form}
\acrodef{CPA}{Correlation Power Analysis}
\acrodef{CPU}{Central Processing Unit}
\acrodef{DPA}{Differential Power Analysis}
\acrodef{DPL}{Dual-rail with Precharge Logic}
\acrodef{FPGA}{Field-Programmable Gate Array}
\acrodef{MDPL}{Masked Dual-rail with Precharge Logic}
\acrodef{NICV}{Normalized Inter-Class Variance}
\acrodef{SNR}{Signal-to-Noise Ratio}
\acrodef{SPA}{Simple Power Analysis}
\acrodef{WDDL}{Wave Dynamic Differential Logic}

\end{document}